\documentclass[conference]{IEEEtran}

\ifCLASSINFOpdf
\else
\fi

\usepackage{epsfig,endnotes,xspace,url,diagbox}
\usepackage[bookmarks=false]{hyperref}
\usepackage[numbers,sort&compress]{natbib}
\usepackage{amsmath,amsthm,amsfonts}

\usepackage{enumerate}
\usepackage{varioref}
\usepackage{graphicx}

\usepackage{caption}
\usepackage{subcaption}
\usepackage{float}
\usepackage{epstopdf}
\usepackage{xcolor}
\usepackage{multirow}
\usepackage{comment}
\usepackage{color}
\usepackage[utf8]{inputenc}
\usepackage{mathtools}
\usepackage{wrapfig}
\usepackage{tabularx}
\usepackage{xspace}
\usepackage[noend, ruled,linesnumbered,vlined]{algorithm2e}
\SetKwComment{Comment}{// }{}
\SetAlgoLined

\usepackage{cases}
\usepackage{stfloats}
\usepackage{fancyhdr}
\usepackage{booktabs}

\usepackage{enumitem}
\setlist[itemize]{leftmargin=*}

\hypersetup{
  colorlinks,
  linkcolor={blue!70!green},
  citecolor={green!70!blue},
  urlcolor={orange!70!red}
}

\newcommand{\etal}{\textit{et al.}\xspace}
\newcommand{\ie}{\textit{i.e.}\xspace}
\newcommand{\eg}{\textit{e.g.}\xspace}

\newcommand{\etc}{\textit{etc.}\xspace}

\newcommand{\mypara}[1]{\noindent\textbf{#1.} \xspace}

\renewcommand{\Pr}[1]{\ensuremath{\mathsf{Pr}\left[#1\right]}\xspace}

\newtheorem{theorem}{Theorem}
\newtheorem{lemma}[theorem]{Lemma}

\newtheorem{definition}{Definition}

\definecolor{revision}{RGB}{0,0,255}

\newcommand{\revstart}{\begin{color}{revision}}
\newcommand{\revend}{~\!\!\end{color}}

\newcommand{\method}{\ensuremath{\mathsf{PrivATE}}\xspace}
\newcommand{\syn}{\ensuremath{\mathsf{PrivSyn}}\xspace}
\newcommand{\aim}{\ensuremath{\mathsf{AIM}}\xspace}
\newcommand{\ipw}{\ensuremath{\mathsf{IPW}}-{\ensuremath{\mathsf{PP}}}\xspace}
\newcommand{\smdp}{\ensuremath{\mathsf{SmoothDPM}}\xspace}
\newcommand{\dpci}{\ensuremath{\mathsf{DPCI}}\xspace}

\usepackage{etoolbox}
\makeatletter
\patchcmd{\hyper@makecurrent}{%
    \ifx\Hy@param\Hy@chapterstring
        \let\Hy@param\Hy@chapapp
    \fi
}{%
    \iftoggle{inappendix}{%
        \@checkappendixparam{chapter}%
        \@checkappendixparam{section}%
        \@checkappendixparam{subsection}%
        \@checkappendixparam{subsubsection}%
        \@checkappendixparam{paragraph}%
        \@checkappendixparam{subparagraph}%
    }{}%
}{}{\errmessage{failed to patch}}

\newcommand*{\@checkappendixparam}[1]{%
    \def\@checkappendixparamtmp{#1}%
    \ifx\Hy@param\@checkappendixparamtmp
        \let\Hy@param\Hy@appendixstring
    \fi
}
\makeatletter

\newtoggle{inappendix}
\togglefalse{inappendix}

\apptocmd{\appendix}{\toggletrue{inappendix}}{}{\errmessage{failed to patch}}

\begin{document}
\title{\method: Differentially Private Average Treatment Effect Estimation for Observational Data}

\author{
}

\author{\IEEEauthorblockN{Quan Yuan\IEEEauthorrefmark{1}\IEEEauthorrefmark{2}\thanks{\IEEEauthorrefmark{2}Quan Yuan’s work on this paper was done while working as a visiting student at the
University of Virginia.},
Xiaochen Li\IEEEauthorrefmark{4},
Linkang Du\IEEEauthorrefmark{5}, 
Min Chen\IEEEauthorrefmark{6},
Mingyang Sun\IEEEauthorrefmark{7}, \\
Yunjun Gao\IEEEauthorrefmark{1}, 
Shibo He\IEEEauthorrefmark{1},
Jiming Chen\IEEEauthorrefmark{1}\IEEEauthorrefmark{8},
Zhikun Zhang\IEEEauthorrefmark{1}\IEEEauthorrefmark{3}\thanks{\IEEEauthorrefmark{3}Zhikun Zhang is the corresponding author.}}

\IEEEauthorblockA{\IEEEauthorrefmark{1}Zhejiang University, \IEEEauthorrefmark{2}University of Virginia,
\IEEEauthorrefmark{4}University of North Carolina at Greensboro,
\IEEEauthorrefmark{5}Xi’an Jiaotong University, \\
\IEEEauthorrefmark{6}Vrije Universiteit Amsterdam, 
\IEEEauthorrefmark{7}Peking University,
\IEEEauthorrefmark{8}Hangzhou Dianzi University\\ 
Email: 
\IEEEauthorrefmark{1}\{yq21, gaoyj, s18he, cjm, zhikun\}@zju.edu.cn, \IEEEauthorrefmark{4}X\_LI12@uncg.edu, 
\IEEEauthorrefmark{5}linkangd@gmail.com, \\ 
\IEEEauthorrefmark{6}m.chen2@vu.nl,
\IEEEauthorrefmark{7}smy@pku.edu.cn
}}

\IEEEoverridecommandlockouts
\makeatletter\def\@IEEEpubidpullup{6.5\baselineskip}\makeatother
\IEEEpubid{\parbox{\columnwidth}{
		Network and Distributed System Security (NDSS) Symposium 2026\\
		23 - 27 February 2026 , San Diego, CA, USA\\
		ISBN 979-8-9919276-8-0\\  
		https://dx.doi.org/10.14722/ndss.2026.241350\\
		www.ndss-symposium.org
}
\hspace{\columnsep}\makebox[\columnwidth]{}}

\maketitle

\begin{abstract}
Causal inference plays a crucial role in scientific research across multiple disciplines.
Estimating causal effects, particularly the average treatment effect (ATE), from observational data has garnered significant attention. 
However, computing the ATE from real-world observational data poses substantial privacy risks to users. 
Differential privacy, which offers strict theoretical guarantees, has emerged as a standard approach for privacy-preserving data analysis.
However, existing differentially private ATE estimation works rely on specific assumptions, provide limited privacy protection, 
or fail to offer comprehensive information protection.

To this end, we introduce~\method, a practical ATE estimation framework that ensures 
differential privacy.
In fact, various scenarios require varying levels of privacy protection.
For example,
only test scores are generally sensitive information in education evaluation,
while all types of medical record data are usually private.
To accommodate different privacy requirements, 
we design two levels (\ie, label-level and sample-level) of privacy protection in~\method.
By deriving an adaptive matching limit, \method effectively balances noise-induced error and matching error, leading to a more accurate estimate of ATE. 
Our evaluation validates the effectiveness of \method.
\method outperforms the baselines on all datasets and privacy budgets.
\end{abstract}

\IEEEpeerreviewmaketitle

\section{Introduction}
\label{sec:introduction}

Causal inference, the process of determining a causal relationship by analyzing the conditions under which an effect occurs, has been a fundamental research focus for decades in various fields~\cite{brand2023recent}, including healthcare~\cite{glass2013causal},
economics~\cite{varian2016causal}, statistics~\cite{pearl2003statistics}, public policy~\cite{matthay2022causal}, education~\cite{cordero2018causal}, \etc
A common example of causal inference is evaluating the impact of taking a drug by analyzing patient data, which can assist doctors in making informed decisions.
There are two typical settings for causal inference: randomized controlled trials (RCTs) and observational studies.
In RCTs, the treatment assignment is controlled by random assignment,
\eg, all patients are randomly assigned to two groups: One group takes the drug, while the other does not.
However, randomized trials are frequently impractical due to ethical, technical, or economic limitations in contexts like studying smoking behavior or assessing economic policies.
In contrast, observational studies (\ie, not intervening in individual grouping, only observing and analyzing naturally occurring
data) are more practical~\cite{shalit2017estimating},
\eg, we can only analyze existing patient data but have no control over whether a patient takes the medicine. 
The setting of observational studies has gained increasing attention due to the abundance of available data and the low budget requirement~\cite{yao2021survey}.

A key task in observational studies is to estimate the average treatment effect (ATE), which quantifies the overall impact of treatment across all samples.
Here, ATE is calculated as the mean of individual treatment effects, where an individual treatment effect is defined as the difference between a sample's potential outcome under treatment and its potential outcome without treatment.
ATE estimation in observational studies often suffers from selection bias and missing information~\cite{berrevoets2023impute}.
There may be significant differences in the characteristics of the treated and control groups. 
In addition, for people who take the drug, the effect of not taking it is unknown.
To mitigate the impact of bias and missing information,
a common solution is to estimate the counterfactual results of each sample and then calculate the causal effect, \eg, sample matching~\cite{rosenbaum1983central}.

However, directly computing ATE from real data using the above approach in observational studies poses significant risks of privacy leakage.
Data utilized in causal inference often contain sensitive personal information~\cite{kusner2016private}.
Direct manipulation and analysis of true data are increasingly challenged by  growing concerns over privacy and the
emergence of regulations for safeguarding individual data.
For example, releasing any statistical information derived from real medical data poses a risk of compromising patient privacy.
For a strong threat model, the attacker could perform a differential attack to infer whether the specific sample’s data is included in the dataset (\ie, comparing query results that include and exclude the sample).

Differential privacy (DP)~\cite{dwork2006calibrating}, a golden standard in the privacy community, has been widely applied for privacy-preserving data analysis~\cite{bittau2017prochlo,rogers2020linkedin,yuan2024privcpm,wang2021continuous}.
By injecting carefully designed noise into the aggregated statistical value, DP can ensure that a single user record has a limited impact on the final output.
Due to the advantages of quantifiable privacy guarantees, high flexibility, and low cost, DP has been deployed by many companies and government agencies~\cite{bittau2017prochlo}.
For instance,
LinkedIn builds Pinot~\cite{rogers2020linkedin}, a DP platform
that enables analysts to gain insight from its members' content engagements.
Although DP serves as an effective privacy protection strategy, within the context of ATE estimation in observational studies, only a small amount of literature has explored privacy-preserving solutions using DP~\cite{lee2019privacy,imbens2015causal,guha2025differentially,ohnishi2024differentially,koga2024differentially}. 

\mypara{Existing Solutions}
Existing studies exhibit several limitations in terms of 
the assumptions of the problem, the scope of protection, and the methodological implementation.
First, 
some approaches assume binary outcomes (\eg, a patient’s medication outcome is categorized as success or failure rather than represented by a continuous numerical indicator), which constrains their applicability.
Second, many solutions offer protection only for partial data attributes (\eg, safeguarding covariates such as patient age and height while leaving medication outcome unprotected). 
Third, most existing works address selection bias through sample reweighting~\cite{rosenbaum1987model},
which aims to
eliminate the distribution differences between
the treated and control groups. 
To ensure a bounded sensitivity, these methods typically employ a fixed, pre-defined truncation threshold to limit individual sample weights. 
However, such a fixed configuration inherently lacks flexibility and interpretability.
Therefore, it is challenging to design a more practical and highly flexible privacy-preserving ATE estimation framework in observational
studies.

\mypara{Our Proposal}
To overcome the limitations of the existing literature and eliminate data bias while effectively protecting user privacy,
we propose a matching-based framework \method 
to estimate the ATE for observational data in a private manner.
Compared to existing works, 
\method provides a more practical and general solution.
In particular, \method does not rely on idealized assumptions such as binary outcomes, which expands its application scenarios.
Furthermore,
considering the different trade-offs between utility and privacy in various scenarios, \method includes two levels of privacy protection: label-level and sample-level.
Label-level protection only perturbs the outcome, which offers higher utility. 
Sample-level protection perturbs all attributes, including treatment, covariates, and outcome, which provides the strongest privacy protection.
The two levels of protection provide solutions for different application scenarios: 
Label-level protection is suited for outcome-sensitive settings (\eg, education evaluation), where other information is publicly available. 
Sample-level protection is for cases involving fully sensitive data (\eg, medical study), where all attributes (including medical records, administered treatments, and outcomes) require protection.

In the causal effect estimation,
it is crucial to restrict the maximum number of matches for each sample, otherwise high sensitivity will occur.
In this way, there will be two types of error in the final ATE estimation: noise error and matching error.
However, it is challenging to choose a suitable matching limit for vairous datasets and privacy budgets.
Taking into account the combined influence of noise error and matching error,
we propose an adaptive matching limit determination mechanism to strike a balance between reducing global sensitivity and improving matching accuracy.
On this basis, we can calculate the counterfactual outcome for each sample in a more accurate manner.
Furthermore, we choose to perturb the sum of aggregated outcomes rather than individual outcomes to reduce the error in the ATE estimation.

\mypara{Evaluation}
We compare~\method with the baseline methods on multiple typical datasets, including real, semi-real, and synthetic datasets.
The experimental results show the superiority of~\method.
For instance,
for the sample-level privacy,
\method consistently outperforms other baseline methods across all datasets and budgets.
In addition, for the label-level privacy,
\method can achieve a low relative error (\ie, less than $0.2$) on multiple datasets even when the privacy budget is $0.5$.
We further verify the effectiveness of our proposed adaptive matching limit determination mechanism with a comparison to the fixed matching limit methods. 
We also explore the impact of the hyperparameter for matching limit calculation.
Moreover, we illustrate the influence of various privacy budget allocations on the ATE estimation.

\mypara{Contributions}
In summary, the main contributions of this paper are four-fold: 
\begin{itemize}[itemsep=2pt,topsep=2pt,parsep=0pt]
    \item We propose \method, a more practical and effective privacy-preserving ATE estimation framework under %
    differential privacy, outperforming existing works.

    \item 
    In~\method,
    we provide two levels of privacy protection (\ie, label-level and sample-level)
    to satisfy different trade-offs between utility and privacy in various scenarios.

    \item 
    We further design an adaptive matching limit determination mechanism to strike a balance between reducing global sensitivity and improving matching accuracy.

     \item We conduct extensive empirical experiments on multiple datasets to illustrate the effectiveness of \method.
     Under the same privacy settings, \method achieves superior performance compared to the baseline methods.
    \method is open-sourced at 
    \url{https://github.com/sec-priv/PrivATE}.

\end{itemize}

\section{Preliminaries}
\label{sec:preliminary}

\subsection{Causal Inference}

In general, a causal inference task estimates how the outcome would change if another treatment had been applied.
Due to the widespread availability of observational data,
estimating treatment effects from such naturally occurring datasets has garnered increasing attention.
Observational data typically includes a group of individuals who have received different treatments, their corresponding outcomes, and possibly additional information, but without direct access to the mechanism or reasons for taking the specific treatment~\cite{imbens2015causal}.

\begin{definition}[Treatment] 
Treatment $T$ represents the action that applies to a sample.
The group of samples with treatment $T=1$ is called the treated group, and the group of samples with $T=0$ is called the control group.
\end{definition}

\begin{definition}[Potential Outcome]
    For each unit-treatment pair, the outcome of that treatment when applied on that sample is the potential outcome.
    The potential outcome of treatment with value $t$ is denoted as $Y(T=t)$.
\end{definition}

\begin{definition}[Observed Outcome]
    The observed outcome is also called factual outcome (denoted as $Y^F$), which represents the outcome of the treatment that is actually applied. 
    $Y^F=Y(T=t)$ where $t$ is the treatment actually applied.
\end{definition}

\begin{definition}[Counterfactual Outcome]
    The counterfactual outcome $Y^{CF}$ is the outcome if the sample took another treatment.
    $Y^{CF}=Y(T=1-t)$ where $t$ is the treatment actually applied.
\end{definition}

\begin{definition}[Covariate]
     Covariate $X$ is the variable that is not affected by the treatment but still affects the experimental results.
\end{definition}

\mypara{Average Treatment Effect}
Average treatment effect (ATE) is defined as follows:
\begin{equation}
\label{eq:ATE_calculation}
    \tau = \mathbb{E}[Y(T=1)-Y(T=0)],
\end{equation}
where $Y(T = 1)$ and $Y(T = 0)$ are the potential treated and control outcomes of the whole population, respectively.

\mypara{Mainstream Solutions for ATE Estimation}
Currently, there are two main methods that can conduct the ATE estimation while mitigating the impacts of bias and missing information.
One way is to eliminate the distribution differences between the treated and control groups, \eg, sample reweighting~\cite{rosenbaum1987model}.
By adjusting the weight of each sample, sample reweighting ensures a similar distribution between the treated and control groups.
However, applying DP to this approach often requires predefined thresholds to limit global sensitivity, which lacks flexibility and interpretability.

The other is to estimate the counterfactual results of each sample and then calculate the causal effect, \eg, sample matching~\cite{rosenbaum1983central}.
This method pairs treated and control samples with similar characteristics. 
This approach can reduce selection bias by identifying individuals with similar characteristics in the treated and control groups, ensuring that the matched samples are as balanced as possible on the covariates.
Given its intuitiveness and practicality,
we choose to achieve a privacy-preserving framework based on matching in this work. 

\mypara{Propensity Score Matching}
As a typical matching method, propensity score matching (PSM) is widely used in observational experiments due to its strong interpretability and low matching complexity~\cite{yao2021survey}.

Therefore, we utilize the PSM approach as a basis to estimate counterfactual results and eliminate the bias of causal effects caused by systematic differences between the treatment and control groups.
The propensity score is defined as the conditional probability of treatment given related variables:
\begin{equation}
    e(x) = \Pr{T=1|X=x}.
\end{equation}

The propensity score reflects the probability of a sample being assigned to the treatment given a series of observed variables. 
However, in most observational studies, the treatment assignment mechanism is unknown.
A common approach is to fit a propensity score function using a standard statistical model on the dataset $D$. 
In this paper, logistic regression is adopted since it is the most frequently used model in existing works.
As a result, on the basis of the absolute value of the difference between various propensity scores, the similarity between any two samples can be calculated and utilized to match.
The distance between the sample $u_1$ in the treated group and the sample $u_2$ in the control group is as follows:
\begin{equation}
    \label{eq:distance_calculation}
    dis(u_1,u_2)=|e(x_1)-e(x_2)|,
\end{equation}
where $e(x_1)$ represents the propensity score of sample $u_1$, and $e(x_2)$ represents the propensity score of sample $u_2$.

The goal of matching is to identify several most similar samples from the opposite treatment group for each sample in the current treatment group. 
Then, the counterfactual outcome can be obtained based on the matched samples.
In general, the counterfactual outcome of the $i$-th sample is as follows:
\begin{equation}
\label{eq:counterfactual_estimation}
    Y^{CF}_i = \frac{1}{|\mathcal{P}(i)|}\sum_{l\in \mathcal{P}(i)}Y^{F}_{l},
\end{equation}
where $\mathcal{P}(i)$ is the matched neighbors of sample $i$ in the opposite treatment group.
Based on the observed and counterfactual outcomes of each sample, ATE can be obtained by~\autoref{eq:ATE_calculation}.

\subsection{Differential Privacy}
\label{section:preliminary_dp}

Differential Privacy (DP)~\cite{dwork2006calibrating} was designed for the data privacy-protection scenarios, where a trusted data curator collects data from individual users, applies perturbation to the aggregated results, and then publishes them.
Intuitively, DP guarantees that any single sample from the dataset has only a limited impact on the output. 

\begin{definition}
[$(\varepsilon,\delta)$-Differential Privacy]
\setlength{\belowdisplayskip}{-0.3cm}
{An algorithm $\mathcal{A}$ satisfies $(\varepsilon,\delta)$-differential privacy ($(\varepsilon,\delta)$-DP), where $\varepsilon>0$, if and only if for any two neighboring datasets $D$ and $D^{\prime}$, 
we have
} 
\begin{equation*}
    \forall O \subseteq Range (\mathcal{A}):\Pr {\mathcal{A}(D) \in O }
    \le e^{\varepsilon} \Pr {\mathcal{A}(D^{\prime}) \in O}+\delta,
\end{equation*}
\label{def:sample_dp}
\end{definition}
\noindent
where Range $(\mathcal{A})$ denotes the set of all possible outputs of the algorithm $\mathcal{A}$,
and $\delta$ indicates the probability of $\mathcal{A}$ failing to satisfy DP.
When $\delta=0$, which is the case we consider in this work (\ie, pure DP),
we write $\varepsilon$-DP for convenience.
Pure DP can provide strict theoretical guarantees, while approximate DP (\ie, $\delta>0$) has a certain probability of violating theoretical constraints.
Approximate DP relaxes the privacy constraint to enable the use of a wider range of composition properties, which may be helpful to improve the utility.
At the same time, according to the experimental results of~\autoref{sec:evaluation}, our method (satisfying pure DP) still shows significant advantages over the baselines (satisfying approximate DP).

In addition,
the definition of $\varepsilon$-sample differential privacy ($\varepsilon$-Sample DP) in the paper is consistent with 
$\varepsilon$-DP.
Here, we consider two datasets $D$ and $D^{\prime}$ to be \emph{neighbors}, denoted as $D\simeq D^{\prime}$, if and only if $D = D^{\prime} + r $ or $ D^{\prime} = D + r$, where $D + r$ is the dataset resulted from
adding the record $r$ to $D$.

\begin{definition}[$(\varepsilon,\delta)$-label Differential Privacy]
    {An algorithm $\mathcal{A}$ satisfies $(\varepsilon,\delta)$-label differential privacy ($(\varepsilon,\delta)$-Label DP), where $\varepsilon>0$, if and only if for any two datasets $H$ and $H^{\prime}$ that differ in the label (observed outcome) of a single sample, we have}
\begin{equation*}
    \forall O \subseteq Range (\mathcal{A}):\Pr {\mathcal{A}(H) \in O }
    \le e^{\varepsilon} \Pr {\mathcal{A}(H^{\prime}) \in O}+\delta.
\end{equation*}
\end{definition}

Similar to Definition~\ref{def:sample_dp},
we consider 
$\delta=0$ in this paper, and we write $\varepsilon$-Label DP instead of $(\varepsilon,\delta)$-Label DP.

\mypara{Laplace Mechanism}
Laplace mechanism (LM) 
satisfies the DP requirements by adding a random Laplace noise to the aggregated results~\cite{dwork2014algorithmic}.
The magnitude of the noise depends on ${GS}_f$, \ie, \emph{global sensitivity}, 
\begin{equation*}
{GS}_f = \max_{D\simeq  D^{\prime} } {\parallel f(D) - f(D^{\prime}) \parallel }_1, 
\end{equation*}
where $f$ represents the aggregation function and $D$ (or $D'$) is the users' data. 
When $f$ outputs a scalar, the Laplace mechanism $\mathcal{A}$ is given below:
\begin{equation*}
    \mathcal{A}_f(D) = f(D) + \mathcal{L}\left(\frac{{GS}_f}{\varepsilon}\right), 
\end{equation*}
where $\mathcal{L}(\beta)$ stands for a random variable sampled from the Laplace distribution $\Pr {\mathcal{L}(\beta)=x}=\frac{1}{2\beta}e^{-\left | x \right | / \beta}$.
When $f$ outputs a vector, $\mathcal{A}$ adds independent samples of $\mathcal{L}(\beta)$ to each element of the vector.
The global sensitivity of all elements is the same value.

\mypara{Random Response Mechanism}
The random response (RR) mechanism can be applied to protect the privacy of binary variables~\cite{wang2016using,zhang2018calm,du2021ahead}.
Given a specific privacy budget $\varepsilon$, the probability of outputting a true binary variable $p$ is as follows:
\begin{equation*}
    p=\frac{e^{\varepsilon}}{e^{\varepsilon}+1}.
\end{equation*}

\mypara{Composition Properties of DP}
The following composition properties of DP are commonly utilized to construct complex differentially private algorithms from simpler subroutines~\cite{dwork2006calibrating}.

\begin{itemize}
    \item \mypara{Sequential Composition}
    Combining multiple subroutines that satisfy differential privacy for
    $\{\varepsilon_{1}, \cdots,\varepsilon_{k}\}$ results in a mechanism which
    satisfies $\varepsilon$-differential privacy for $\varepsilon = { \sum_{i}\varepsilon_{i}} $.
    
    \item \mypara{Parallel Composition}
    Given $k$ algorithms working on disjoint subsets of the dataset, each satisfying DP for $\{\varepsilon_{1},\cdots,\varepsilon_{k}\}$, the result satisfies $\varepsilon$-differential privacy for $\varepsilon=\max\{\varepsilon_{i}\}$.
    
    \item \mypara{Post-processing}
    Given an $\varepsilon$-DP algorithm $\mathcal{A}$, releasing $z(\mathcal{A}(D))$ for any $z$ still satisfies $\varepsilon$-DP, \ie, post-processing the output of a differentially private algorithm does not incur any additional loss of privacy. 
\end{itemize}

\section{Methodology}

\subsection{Problem Definition}
\label{sec:problem_definition}

In this paper, we consider a dataset $D=(T,X,Y)$ that contains multiple dimensions, where $T$ stands for the treatment,
$X$ stands for the related covariates, and $Y$ stands for the observed outcome.
Note that both $T$ and $Y$ contain only one column, while $X$ can contain multiple columns.
Without loss of generality, we assume that $T=0/1$, $X\in[0,1]^d$ ($d$ is the dimension of covariates), and the maximum variation range of outcome is $B$. 
Our goal is to estimate an average treatment effect that closely aligns with the result obtained through propensity score matching while ensuring strict differential privacy.
Specifically, we aim to achieve two levels of privacy protection, \ie, \emph{label-level} and \emph{sample-level}.
For label-level privacy, only the observed outcome $Y$ is private.
For sample-level privacy, all types of data are private.
For ease of reading, 
we summarize the frequently used 
notations in \autoref{table:math_notations}.

\begin{table}[!t]
    \centering
    \caption{Summary of mathematical notations.}
    \label{table:math_notations}
    \vspace{-0.1cm}
    \footnotesize
    \setlength{\tabcolsep}{0.9em}
	\begin{tabular}{cc}
		\toprule
		\textbf{Notation} & \textbf{Description}  \\
		\midrule
		$D$ & %
        Dataset \\
        $T$ & The treatment   \\
	$X$ & The          covariates      \\
    $Y$ & The observed        outcome      \\
    $Y_1$ & The potential treated outcome of the whole population     \\
    $Y_0$ & The potential control outcome of the whole population    \\
    $\varepsilon$ & Privacy budget  \\
        $n$ & The number of all samples \\ 
        $n_t$ & The number of samples in the treated group \\
        $n_c$ & The number of samples in the control group \\
        $d$ & The dimensions of covariates  \\
        $k$ & The value of matching limit \\
		$N$ & The number of neighbors for each sample in the matching \\
		${B}$ & The maximum variation range of outcome  \\
		$w$ & The weights of the regression model \\
		
		$e$ &  The propensity score \\
		$\tau$ & The average treatment effect estimate \\

		\bottomrule
	\end{tabular}
\end{table}

\label{sec:methodology}

\subsection{Motivation}

When implementing propensity score matching under DP, a common idea is to directly add noise to the ATE estimate to achieve DP. 
However, the impact of adding or deleting any sample on the propensity score matching is difficult to evaluate.
Therefore, 
we choose to apply DP to each phase of propensity score matching, thus ensuring the privacy of the entire process.
If ATE is estimated completely according to the matching results, it could introduce excessive noise. 
This will make some samples match too many times, making the sensitivity too high. 
If the number of matches is too low, the estimation of ATE will be inaccurate.
Thus, we determine the maximum number of matches for each sample by estimating the combined impact of noise injection and matching accuracy, thereby achieving a great trade-off between these two aspects.
Considering the privacy requirements in various scenarios in practice, we designed two different levels of privacy protection approaches to estimate ATE, \ie, label-level privacy and sample-level privacy.

Here, we summarize the key challenges of ATE estimation under the two privacy settings as follows:
For label-level privacy, directly applying a standard DP mechanism to existing ATE estimation methods often introduces excessive noise. 
Furthermore, using a fixed matching upper limit is unsuitable across different privacy budgets and data distributions, making it difficult to adaptively determine a limit that balances matching accuracy and privacy protection.
For sample-level privacy, which protects the entire dataset rather than just labels, similar issues arise. 
Under stricter privacy requirements, additional challenges include allocating the overall privacy budget and determining an appropriate matching limit while ensuring DP guarantees at each step.

\subsection{Overview}
\label{section:overview}

\begin{figure*}[htbp]
\centering
\includegraphics[width=0.95\textwidth]{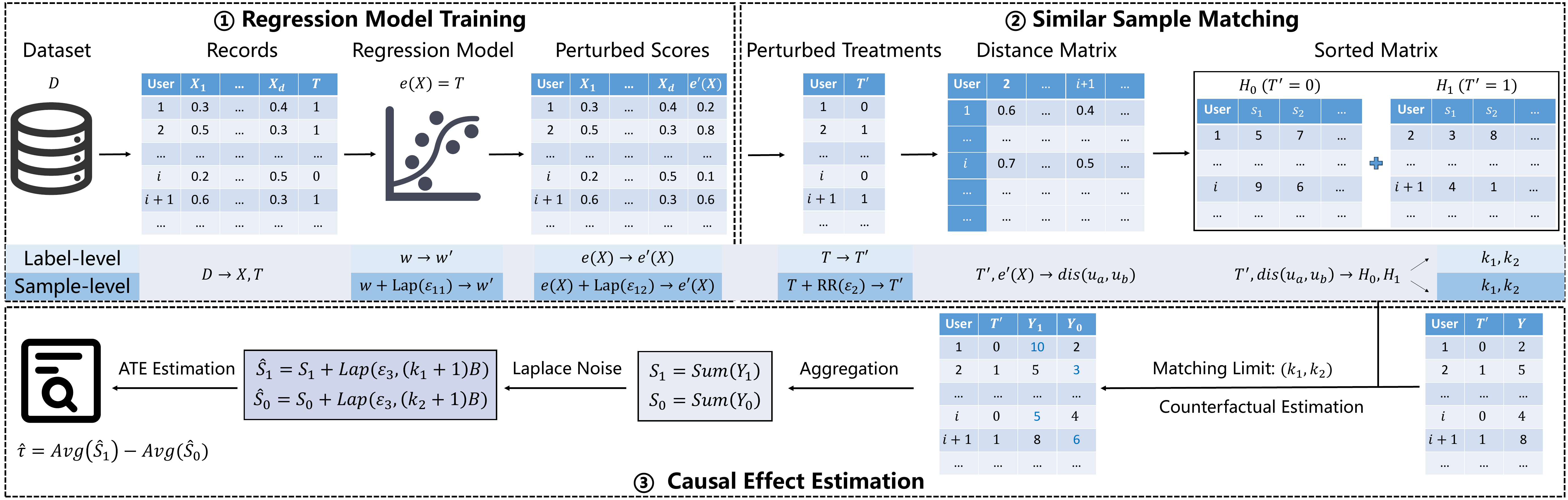}
\caption{
\method overview. \method consists of three phases: Regression model training, similar sample matching, and causal effect estimation.
First, a regression model for calculating propensity scores can be obtained in the regression model training phase.
Then, \method finds the closest neighbors in the opposite group for each sample in the similar sample matching phase.
In the causal effect estimation phase,
\method calculates each sample's counterfactual outcome based on the matching results and the matching limit, \ie, the maximum number of matched of each sample.
After that, the potential outcomes for the control and treated groups are aggregated and perturbed.
Finally, the average treatment effect can be estimated by the %
perturbed outcomes.
}
\label{fig:framework}
\end{figure*}

As shown in~\autoref{fig:framework}, %
the framework of \method contains three phases, \ie, regression model training, similar sample matching, and causal effect estimation. 

\mypara{Phase 1: Regression Model Training (RMT)}
We train a logistic regression model to estimate the propensity scores of all samples.
In the label-level setting, 
the model training and the propensity score calculation do not need to consume the privacy budget. 
The reason is that this process does not require access to the observed outcomes of the samples. 
In the sample-level setting, the training of the regression model and the estimation of the propensity score need to be perturbed to meet DP.

\mypara{Phase 2: Similar Sample Matching (SSM)}
In this phase, the distance between the propensity score of any sample and the propensity scores of all samples in the opposite treatment group can be calculated.
By sorting these scores, we can obtain the most similar neighbors of each sample in the opposite group.
Note that two sorted matrices are calculated here, one for the control group and the other for the treated group.
In the label-level setting,
this procedure still does not visit the observed outcome, thus consuming no privacy budget.
In the sample-level setting,
the true treatment $T$ is perturbed to satisfy DP in this phase.

\mypara{Phase 3: Causal Effect Estimation (CEE)}
Based on the sorted matrices in the last phase,
we can find the closest neighbors in the opposite group for each sample.
Then, the counterfactual outcomes of all samples can be estimated.
Here, we limit the maximum number of times each sample can be used for matching.
The matching limit can be adaptively adjusted based on the privacy budget and the characteristics of the dataset.
After calculating the counterfactual outcomes,
we aggregate and perturb the sum of potential outcomes of all samples.
Then, the ATE can be obtained by~\autoref{eq:ATE_calculation}.

\subsection{Regression Model Training}
\label{section:regression_model_training}

\begin{algorithm}[!t]

        \caption{Regression Model Training (Phase 1)}
        
        \label{algorithm:phase_1}
        \KwIn{Original dataset $D$, 
        privacy level $l$,
        privacy budget $\varepsilon_{11}, \varepsilon_{12}$ (if $l$ is sample-level)
        }
        \KwOut{Propensity Score $e'(X)$}

        Train a logistic regression model based on the covariates $X$ and treatment $T$. \\
        \If{l is label-level}
        {
        $w'\leftarrow w$
        }
        \Else
        {
        \Comment{Private model training (sample-level)}
        $w' \leftarrow w + Lap(\varepsilon_{11},\Delta f_w)$ %
        }
        Calculate the propensity score of each sample $e(X)$ based on the model parameter $w'$. \\
        \If{l is label-level}
        {
        $e'(X)\leftarrow e(X)$
        }
        \Else
        {
        \Comment{Private score calculation (sample-level)}
        $e'(X) \leftarrow e(X) + Lap(\varepsilon_{12},\Delta f_e)$ %
        }

\end{algorithm}

In the first phase,
we train a logistic regression model based on the original dataset.
The covariates $X$ is the independent variable of the regression model, while the treatment $T$ is the dependent variable.
\autoref{algorithm:phase_1} illustrates the basic process of the first phase.

\mypara{Label-level Privacy}
In the label-level privacy, 
only the outcomes need to be protected.
The regression model training phase does not require access to the true outcomes.
Therefore, we can utilize the true parameters of the trained model $w$ to predict the propensity scores of all samples without consuming the privacy budget.
The predicted propensity scores $e(X)$ also do not require to be perturbed in the label-level privacy setting.

\mypara{Sample-level Privacy}
In contrast, sample-level privacy requires to protect the privacy of all types of variables.
If we directly adopt the weight $w$ of the unprotected regression model,
there will be a risk of privacy leakage~\cite{wei2023dpmlbench}.
Therefore, we choose to perturb the training
of the regression model to satisfy DP. 
The training of the logistic
regression model can be regarded as a specific case for regularized
empirical risk minimization. For the logistic regression with $\ell_2$
regularization penalty, the regularized empirical loss can be written
as follows:

\begin{equation}
    J(w)=\frac{1}{n}\sum_{i=1}^n \log(1+e^{-X_i^T w t_i}) + \frac{\lambda}{2}||w||^2_2,
\end{equation}
where $X$ is the training feature (covariates) containing $d$-dim and $t_i$ is the $i$-th sample's treatment.
The weight $w$ can be perturbed to satisfy DP.
The $L_1$-sensitivity of $w$ is $\frac{2{d}}{n\lambda}$, and
the detailed derivation is given in Proof 1 of~\autoref{proof:user_level_DP}.
By injecting Laplace noise into the true weight $w$ with the privacy budget of $\varepsilon_{11}$,
we can generate a privacy-preserving regression model.

After finishing the private model training,
we can utilize the model to calculate the propensity score of each sample.
Since this step requires accessing the true covariate again, we need to add Laplace noise to the relevant query results to meet differential privacy.
The output range of logistic regression model is $[0,1]$, thus the sensitivity of propensity score is $\Delta f_e = 1$. 

In this phase, we inject noise into $w$ and $e(X)$, respectively.
Note that both parts of noise are indispensable.
The purpose of adding noise to $w$ is to protect the privacy of training data (\ie, $X$ and $T$). If the DP regression model is queried using public or non-sensitive data, no additional privacy budget is consumed due to the post-processing property. 
However, $e(X)$ is computed using the actual data $X$, which constitutes additional access to private information. 
To preserve the privacy of $X$, we still need to inject noise into $e(X)$.
In addition, if $w$ is not perturbed and only $e(X)$ is perturbed, we cannot apply the parallel composition to perturb each sample in $X$ because $w$ contains sensitive information.
Therefore, it is necessary to inject noise into $w$ and $e(X)$.

\subsection{Similar Sample Matching}
\label{section:similar_sample_matching}

\begin{algorithm}[!t]

        \caption{Similar Sample Matching (Phase 2)}
        
        \label{algorithm:phase_2}
        \KwIn{Original dataset $D$,
        propensity score $e'(X)$,
        privacy level $l$,
        privacy budget $\varepsilon_{2}$ (if $l$ is sample-level)
        }
        \KwOut{Treatment $T'$, Sorted Matrices $H$}

        \If{l is label-level}
        {
        $T'\leftarrow T$
        }
        \Else
        {
        \Comment{Treatment perturbation (sample-level)}
        $T' \leftarrow RR(T;\varepsilon_2)$ %
        }
        Obtain the division of treated and control groups by $T'$ \\
        \Comment{Distance sorting}
        \For{each sample j in the control group}
        {
        Calculate the distance vector $dis_{j,s_t}$ between $j$ and the samples $s_t$ in the treated group based on~\autoref{eq:distance_calculation} \\
        $H_0^j\leftarrow argsort(dis_{j,s_t})$ 
        }
        \For{each sample j in the treated group}
        {
        Calculate the distance vector $dis_{j,s_c}$ between $j$ and the samples $s_c$ in the control group based on~\autoref{eq:distance_calculation} \\
        $H_1^j\leftarrow argsort(dis_{j,s_c})$ 
        }

\end{algorithm}

In the second phase,
we try to calculate the similarity of each sample with all samples in the other group and rank them.
\autoref{algorithm:phase_2} provides the specific procedures of the similar sample matching.

\mypara{Label-level Privacy}
In this setting,
the treatment $T$ is accessible, which does not need to be perturbed.
Then,
we can calculate the distance between each sample and all samples in the opposite treatment group based on the propensity score $e'(X)$ obtained from the first phase.
The specific calculation formula is shown in~\autoref{eq:distance_calculation}.
Then, we traverse each sample and sort all candidate samples in the opposite treatment group in ascending order based on the distance vectors.
From this, we can obtain two sorted index matrices, one for the control group (\ie, $H_0$) and the other one for the treated group (\ie, $H_1$).
The sorted matrices can be utilized for counterfactual estimation in the next phase.

\mypara{Sample-level Privacy}
Unlike label-level privacy,
the treatment $T$ is sensitive information in the sample-level setting.
Considering that $T$ is a binary variable, it is not appropriate to inject Laplace noise, 
which is used for continuous variables. 
Here, based on the privacy budget of $\varepsilon_2$, we adopt the random response mechanism to perturb $T$, which effectively protects privacy and improves data utility.
The next steps are similar to the label-level settings.
Using the perturbed treatment $T'$ and perturbed propensity scores $e'(X)$,
we can calculate the distance between each sample and other samples in the opposite perturbed treatment group.
Then, we traverse each sample and sort other samples 
according to the corresponding distance values.
After that, we obtain two sorted matching matrices.

\subsection{Causal Effect Estimation}
\label{section:causal_effect_estimation}

\begin{algorithm}[!t]

        \caption{Causal Effect Estimation (Phase 3)}
        
        \label{algorithm:phase_3}
        \KwIn{Original dataset $D$, sorted matrices $H$,
        the number of neighbors in matching $N$,
        treatments $T'$,
        privacy level $l$,
        the range of outcome $B$,
        privacy budget $\varepsilon_{3}$,
        the error coefficient $c$ (label-level) or $h$ (sample-level)
        }
        \KwOut{Average treatment effect estimate $\hat \tau$}

        Count the maximum number of times that all samples appear in the first $N$ neighbors of the sorted index matrices $M$ (label-level) or $M'$ (sample-level). \\ 
        \Comment{Obtain the number of samples in the treated and control groups, and the average maximum number of matches.}
        \If{l is label-level}
        {
        $n_t, n_c \leftarrow T, M_1=\frac{M}{N}$
        }
        \Else{
        $n'_t, n'_c \leftarrow T',M'_1=\frac{M'}{N}$ // Sample-level
        }
        $r_1=\frac{n_t}{n_c}$ (or $r_1=\frac{n'_t}{n'_c}$) \\
        
        \Comment{Matching limit calculation}
        \If{$r_1 \leq 1$}
        {
        Calculate the matching limit of treated group $k_1$ based on Equations
        \ref{eq:cal_opt_k}-\ref{eq:round_clip_opt_k}
        (or Equations \ref{eq:user_level_opt_k}-\ref{eq:user_level_round_clip_opt_k}) \\ 
        $k_2=\max(1,round(k_1\cdot r_1))$ 
        }
        \Else{
        Calculate the matching limit of control group $k_2$ based on Equations
        \ref{eq:cal_opt_k}-\ref{eq:round_clip_opt_k}
        (or Equations \ref{eq:user_level_opt_k}-\ref{eq:user_level_round_clip_opt_k}) \\ 
        $k_1=\max(1,round(k_2/r_1))$
        
        }
        \Comment{Counterfactual estimation}
        \For{i-th sample in D}
        {
        Remove candidate samples that have reached the upper matching limit from the original sorted list. \\
        Select the nearest $N$ neighbors from the remaining sorted list, and add $1$ to their matching counts. \\
        Calculate the counterfactual outcome based on~\autoref{eq:counterfactual_estimation}.
        }
        \Comment{Outcome aggregation}
        $S_1=Sum(Y_1)$, $S_0=Sum(Y_0)$ \\
        \Comment{Noise perturbtion}
        $\hat{S}_1=S_1+Lap(\varepsilon_3,(k_1+1)B)$,
        $\hat{S}_0=S_0+Lap(\varepsilon_3,(k_2+1)B)$ \\
        \Comment{ATE estimation}
        $\hat\tau=\frac{1}{n}\cdot\hat{S}_1-\frac{1}{n}\cdot\hat{S}_0$

\end{algorithm}

In the third phase, we need to calculate each sample’s counterfactual estimate and obtain the final ATE estimate.

\mypara{Label-level Privacy}
If we apply the original non-private approach to select $N$ neighbors for computing counterfactual results,
the number of matches for some samples may be extremely high, which will make the global sensitivity large.
On the other hand, if we set the upper limit of the number of times each sample is matched very small,
the disturbance intensity of the noise will be reduced,
but the error of the counterfactual estimate will tend to increase.
In addition, the noise intensity under various privacy budgets is different, making matching selection more challenging.

To address the above difficulties, we design an adaptive matching upper limit determination mechanism by considering the combined impact of noise injection and matching error, which can provide different matching upper limits for various privacy budgets.

We first consider the expected squared error of estimating an aggregated potential outcome $S=Sum(Y)$. 
Assuming that $\hat S$ is the estimation of $S$, 
then the expected squared error can be written as the summation of variance and the squared bias of $\hat S$:
\begin{equation}
    \label{eq:expected_squared_error}
    \mathbb{E}[(\hat S - S)^2]=\mathsf{Var}[\hat S]+\mathsf{Bias}[\hat S]^2
\end{equation}

Given the maximum variation range of the outcome $B$,
the matching upper limit for each sample $k$ and privacy budget $\varepsilon$,
we can obtain $\mathsf{Var}[\hat S]\approx\frac{2k^2B^2}{\varepsilon^2}$.

For $\mathsf{\hat S}$, its value is related to the number of samples, the true maximum number of matches, and the set matching upper limit.
We can count the maximum number of times that all samples appear in the first $N$ neighbors of the sorted index matrices, denoted as $M$.
Here, we let $M_1=\frac{M}{N}$, which represents the average maximum number of matches.
Intuitively, the smaller the matching upper bound, the greater the bias.
The larger the number of samples and the average maximum number of matches, the larger the bias.
Therefore, we estimate $\mathsf{\hat S}$ as follows:

\begin{equation}
    \label{eq:appro_bias}
    \mathsf{Bias}[\hat S]\approx c\cdot B\cdot n_1\cdot \frac{M_1}{k},
\end{equation}
where $n_1=\max{(n_t,n_c)}$ is the number of samples in the treated group or control group and $c$ is the error coefficient, which is a hyperparameter. 
Therefore, the combined error of noise perturbation and matching limit can be approximately estimated as follows:
\begin{equation}
    \label{eq:combined_appro_error}
    \mathbb{E}[(\hat S - S)^2]\approx \frac{2k^2B^2}{\varepsilon^2} + c^2\cdot B^2\cdot n_1^2\cdot \frac{{M_1}^2}{k^2}.
\end{equation}

By calculating the minimum value of~\autoref{eq:combined_appro_error},
we can obtain the optimal value $k^*$ as follows:
\begin{equation}
\label{eq:cal_opt_k}
    k^*=\sqrt{\frac{\varepsilon \cdot c \cdot n_1 \cdot M_1}{2}}.
\end{equation}

Since the matching limit is a positive integer and does not require to exceed the true average maximum number of matches $M_1$,
we can obtain the final optimal matching limit $k_f$ as follows:
\begin{equation}
    \label{eq:round_clip_opt_k}
    k_f=\min( \max( round(k^*),1 ),M_1).
\end{equation}

Here, $k_f$ represents the upper bound of the number of matches for the larger number of treatment groups.
The matching upper bound for the other group can be calculated using $k_f$ and the number of the two treatment groups.
For instance,
if the number of control group $n_c$ is larger than the number of treated group $n_t$ (\ie, $r_1=\frac{n_t}{n_c}\leq 1$),
the number of matches for the treated group will usually be higher than the number of matches for the control group.
In this case,
we obtain
$k_1 = k_f$ and
$k_2 = \max(1,round(k_1 \cdot r_1))$.
On the contrary,
if the number of control group $n_c$ is smaller, we let
$k_2 = k_f$ and
$k_1 = \max(1,round(k_2 /r_1))$.
Since the upper limit is calculated based on the average maximum number of matches, the final matching limit needs to be multiplied by the preset number of neighbors $N$.
The matching limit for treated group is $k_1\cdot N$, and the matching limit for control group is $k_2\cdot N$.

After determining the upper limit of the matching, we traverse all samples and find the closest $N$ neighbors for each sample.
At the beginning, the records of the number of matches for all samples are initialized to $0$.
When $N$ neighbors are selected in each traversal, these $N$ samples' matching counts are increased by $1$.
Samples that reach the upper limit will not be selected in subsequent matches.
For each sample, the counterfactual outcome can be computed by the matched neighbors' observed outcome, as shown in~\autoref{eq:counterfactual_estimation}.

After calculating the counterfactual outcomes of all samples, we can estimate the treatment effect of each sample. However, calculating and perturbing the treatment effect of each individual will introduce a lot of noise, making the average treatment effect estimate inaccurate. 
Therefore, we choose to aggregate the potential outcomes of all samples and add Laplace noise, which can effectively reduce the impact of noise.

As shown in~\autoref{fig:framework},
$Y_1$ is composed of the outcomes of the treated group, while $Y_0$ is composed of the outcomes of the control group.
Laplace mechanism is applied to protect the samples' privacy.
Note that the samples of the treated group and the control group are non-overlapping,
thus these two parts can share the same privacy budget.
Regarding the global sensitivity,
the sensitivity of treated group is $(k_1+1)B$, which is determined by two factors: the matching upper limit of the counterfactual estimation and its own observed outcome.
Similarly, we can obtain the sensitivity of control group is $(k_2+1)B$.
Based on the perturbed aggregated outcomes,
we can obtain the final ATE estimate $\hat\tau$.

\mypara{Sample-level Privacy}
Considering the impact of noise injection and matching error, sample-level also needs to determine a
suitable maximum number of matches for each sample to achieve a
promising estimation result,
which is similar to label-level privacy.
However, since the treatment of each sample and the matching results are perturbed,
it is difficult to accurately estimate the errors caused by noise and matching. 
According to the source of the error, the setting of an ideal match upper limit is related to the privacy budget and the true maximum number of matches, as well as the characteristics of the dataset.
Unfortunately, most of this information cannot be obtained in the sample-level setting.
Inspired by~\autoref{eq:cal_opt_k} in the label-level privacy,
we set the value of matching limit in the sample-level privacy as follows:
\begin{equation}
    \label{eq:user_level_opt_k}
    k^*=\sqrt{\frac{\varepsilon_3 \cdot h \cdot n'_1 \cdot M'_1}{2}},
\end{equation}
where $\varepsilon_3$ is the privacy budget used for perturbing the aggregated outcomes, $h$ is the error coefficient, and $n'_1=\max{(n'_t,n'_c)}$ is the number of samples in the perturbed treated group or control group.
$M'_1$ is the average maximum number of matches.
The calculation of $M'_1$ is similar to $M_1$ in the label-level privacy.
The only difference is that $M'_1$ is calculated based on the perturbation information rather than true information.
When the privacy budget is high, the noise intensity is low, and increasing $k^*$ helps reduce the matching error.
If $n$ is high, the true matching upper limit is usually higher, which requires a higher $k^*$.
Since the calculation result of~\autoref{eq:user_level_opt_k} may not be an integer, we further process $k^*$ as follows:
\begin{equation}
\label{eq:user_level_round_clip_opt_k}
    k_f=\max(round(k^*),1 )
\end{equation}

After calculating the matching limit of one group, the maximum matching upper bound of the other group can be obtained based on the number of samples in the perturbed groups.
Next, we can calculate the counterfactual estimate for each sample.
We then aggregate the potential outcomes of $T'=1$ and $T'=0$ for all samples.
To satisfy DP,
we utilize Laplace noise to perturb the aggregated outcomes, with the privacy budget of $\varepsilon_3$.
Finally,
we can compute the ATE estimate according to the perturbed aggregated outcomes, which is similar to the calculation at label-level.

\subsection{Putting Things Together}
\label{section:putting_things_together}

The above three phases constitute the overall process of~\method.
Due to space limitations,
we defer the pseudo-code to \autoref{sec:appendix_overall_workflow}.

\subsection{Algorithm Analysis}
\label{section:algorithm_analysis}

\mypara{Privacy Analysis}
Recalling~\autoref{fig:end_to_end},
\method mainly consists of three phases: %
regression model training, similar sample sampling, and causal effect estimation.
For the label-level privacy, 
the outcome is visited and perturbed in the causal effect estimation phase, with the privacy budget of $\varepsilon$.

For the sample-level privacy of~\method,
the total privacy budget is divided into all phases.
In the phase of regression model training,
\method needs to protect the true regression model weights and the true propensity score, 
which consume privacy budget of $\varepsilon_{11}$ and $\varepsilon_{12}$ respectively.
In the second phase,
the treatment is perturbed based on the privacy budget $\varepsilon_2$.
In the causal effect estimation phase,
the aggregated outcome consumes the privacy budget $\varepsilon_3$.
Therefore, the total privacy budget is $\varepsilon=\varepsilon_{11}+\varepsilon_{12}+\varepsilon_2+\varepsilon_3$.
We obtain the following theorems, and the detailed proofs are deferred to \autoref{proof:user_level_DP}.

\begin{theorem}
\label{throrem:label_level}
If $l$ is label-level, \autoref{algorithm:put_together} satisfies $\varepsilon$-Label~DP.
\end{theorem}

\begin{proof}
    (Sketch) 
    In the ``regression model training'' and ``similar sample matching'' phases,
    no privacy budget needs to be consumed since the true observed outcome is not visited.
In the ``causal effect estimation'' phase,
the aggregated outcome is perturbed by Laplace mechanism with the privacy budget $\varepsilon$.
Therefore, \autoref{algorithm:put_together} satisfies $\varepsilon$-Label DP.

\end{proof}

\begin{theorem}
\label{throrem:user_level}
If $l$ is sample-level, 
\autoref{algorithm:put_together} satisfies $\varepsilon$-Sample~DP, where $\varepsilon=\varepsilon_{11}+\varepsilon_{12}+\varepsilon_2+\varepsilon_3$.
\end{theorem}

\begin{proof}
    (Sketch) In the sample-level setting, all types of data are sensitive.
    In the first phase, both the 
    model training and propensity score calculation use real information.
    Therefore, the model weights (consuming the privacy budget $\varepsilon_{11}$) and the propensity scores (consuming  $\varepsilon_{12}$) are injected with Laplace noise to achieve DP.
    In the second phase,
    the true treatment is utilized to guide the sample matching. To ensure DP, the treatment is perturbed with the privacy budget $\varepsilon_2$.
    The distance matrix calculation and sorting are finished based on the perturbed information, which is regarded as post-processing.
    In the third phase, 
    the outcome is perturbed based on the privacy budget $\varepsilon_3$, which is similar to the label-level setting.
    According to the sequential composition, \autoref{algorithm:put_together} satisfies $\varepsilon$-Sample~DP, where $\varepsilon=\varepsilon_{11}+\varepsilon_{12}+\varepsilon_2+\varepsilon_3$.
\end{proof}

\mypara{Error Analysis}
For label-level privacy,
we theoretically analyze the error bound of the aggregated potential outcome in~\autoref{throrem:error_bound_label}.
The detailed proof is in~\autoref{proof_error_bound_label}.

\begin{theorem}
    \label{throrem:error_bound_label}
     For the label-level privacy, the expected squared error of aggregated potential outcome $S$ is bounded by $2(\frac{(k+1)B}{\varepsilon})^2+(\frac{R}{N}B)^2$,
     where $R$ is the total number of times that neighbor samples are replaced when matching without matching upper limit and matching with matching upper limit.
\end{theorem}

\begin{proof}
    (Sketch) Let $\hat S$ denotes the estimation of $S$,
    the expected squared error can be written as the summation of variance and the squared bias of $\hat S$ according to~\autoref{eq:expected_squared_error}.
    The variance part comes from Laplace noise, and the expected value is $2(\frac{(k+1)B}{\varepsilon})^2$.
    The bias part comes from the matching difference caused by whether the matching upper limit is applied, and the upper bound is $(\frac{R}{N}B)^2$.
    Combining the above results, we can obtain the final error bound.
\end{proof}
For sample-level privacy,
the treatment of each sample is perturbed to satisfy DP, making the sample
grouping of the original and privacy-preserving data inconsistent. Furthermore, the regression model is
also perturbed, resulting in different matching results for original and privacy-preserving
settings. 
Therefore, it is difficult to directly derive the error bound.

\mypara{Complexity Analysis}
We compare the time complexity and the space complexity of various methods~\cite{lee2019privacy,koga2024differentially,mckenna2022aim,zhang2021privsyn,schroder2025private}.
The running time of~\method is significantly lower than \aim and \syn.
The space consumption of \method is the lowest, and the space consumption of~\aim is higher than that of other methods.
The detailed analysis can be found in~\autoref{app:complexity_analysis}.

\section{Evaluation}
\label{sec:evaluation}
In this section, we first conduct an end-to-end experiment to illustrate the effectiveness of~\method in~\autoref{subsec:end_to_end}. 
Then, we conduct a hyperparameter study for label-level privacy in~\autoref{section:paramter_study_label_level}.
Furthermore, we explore the impact of hyperparameter for sample-level privacy in~\autoref{section:paramter_study_sample_level}.

\subsection{Experimental Setup}
\label{subsec:experimental_setup}

\mypara{Datasets}
We run experiments on the four typical datasets, including real, semi-real, and synthetic datasets. 
These datasets are classic benchmarks in causal inference and are widely adopted in existing studies~\cite{lee2019privacy,koga2024differentially}.
The basic information of these four datasets are shown in~\autoref{table:dataset_statistics}, and the details of these datasets are deferred to~\autoref{appendix_datasets}.

\mypara{Metric}
To evaluate the quality of various methods~\cite{lee2019privacy,koga2024differentially,mckenna2022aim,zhang2021privsyn}, we utilize the metric of relative error (RE) to show their performance.
The related formula is as follows:

\begin{equation*}
    \label{eq:relative_error}
    RE_{ATE}=\frac{|\hat\tau- \tau|}{\tau},
\end{equation*}
where $\tau$ is the true ATE estimate based on the non-private PSM method and $\hat \tau$ is the perturbed ATE estimate based on the privacy-preserving mechanism.
The RE of the non-private ATE estimate is 0. 
A value of RE closer to 0 indicates a more accurate ATE estimate.

\begin{table}[!t]
    \caption{Dataset Statistics.}
    \vspace{-0.2cm}
    \centering
    \setlength{\tabcolsep}{0.7em}
    \begin{tabular}{c | c | c | c | c}
    \toprule
     \textbf{Datasets} & \textbf{Treated} & \textbf{Control} & \textbf{Total} & \textbf{Type} \\
     \midrule
       IHDP~\cite{hill2011bayesian}  & 139 & 608 & 747 
       & Semi-real
       \\
       Lalonde~\cite{dehejia1999causal}  & 185 & 260 & 445 
       & Real 
       \\
       ACIC~\cite{dorie2019automated} & 858 & 3944 & 4802 & Semi-real  \\
       Synth~\cite{imbens2015causal} & 489 & 511 & 1000 & Synthetic  \\ 
      \bottomrule
    \end{tabular}
    
    \label{table:dataset_statistics}
\end{table}

\mypara{Competitors} 
In this work, we compare \method with two representative approaches.
The first is the existing differential private ATE estimation methods (\ie, \ipw~\cite{lee2019privacy},  \smdp~\cite{koga2024differentially}, and \dpci~\cite{schroder2025private}),
which are most comparable to our work in terms of problem assumptions and privacy protection scope.
\ipw first uses a subset of the original dataset to learn a perturbed propensity score function, and then estimates 
causal effect on the remaining samples using privacy-preserving inverse probability weighting. 
\smdp employs a smooth sensitivity-based mechanism combined with an exact matching estimator, where the matching variables are required to be discrete, to achieve privacy-preserving ATE estimation.
\dpci estimates the ATE through doubly robust estimation,
and guarantee %
DP by output perturbation. 
This method further estimates the differentially private variance and constructs the confidence intervals (CIs),
which is beyond the focus of this paper.
To ensure a fair comparison, we allocate all privacy budgets to the ATE estimation.

The second is the advanced differentially private data synthesis methods (\ie, \syn~\cite{zhang2021privsyn} and \aim~\cite{mckenna2022aim}),
which are another potential solution to the ATE estimation problem.
Including this type of comparison helps better understand the performance of general DP synthesis schemes on this problem.
By effectively capturing the correlation between various attributes of the original dataset and restoring the original distribution as much as possible, 
\syn achieves great data synthesis performance.
By following the select-measure-generate paradigm,
combined with an iterative and greedy approach to select the most useful queries, \aim can achieve low errors 
across a range of experimental settings compared to existing privacy-preserving data synthesis mechanisms.

Note that the above baselines guarantee $(\varepsilon,\delta)$-Sample DP, while our proposed \method can provide stricter $\varepsilon$-DP.

\mypara{Experimental Settings}
For the label-level setting of \method, we set the error coefficient $c=0.01$ in the matching limit calculation. 
Moreover, we have a further discussion about the choice of $c$ in~\autoref{section:paramter_study_label_level}. 
For the sample-level setting of \method,
we set the coefficient $h=0.001$ in the matching limit calculation.
We also provide the impact of various $h$ in~\autoref{section:paramter_study_sample_level}.
For the allocation of privacy budget in sample-level setting, we set $\varepsilon_{11} = \varepsilon_{12} = 0.5 \varepsilon_1$, and $\varepsilon_1:\varepsilon_2:\varepsilon_3=0.1:0.7:0.2$.
We also explore the impact of different privacy budget allocation.

\mypara{Implementation}
We set the total privacy budget $\varepsilon$ ranges from $0.5$ to $4.0$. 
Regarding the number of neighbors $N$ in the counterfactual estimation,
we set $N=5$ 
in the experiments.
In addition,
we also provide the results of $N=\{1,3,7\}$ in~\autoref{subsec:appendix_end_to_end}.
We implement \method with Python 3.8, and all experiments are conducted on a server with Intel(R) Core(TM) i7-11700K @ 3.60GHz and 128GB memory.
We repeat experiment 10 times for each settings, and provide the mean and the standard variance.

\subsection{End-to-End Evaluation}
\label{subsec:end_to_end}

\begin{figure*}[!t]
\captionsetup[ploture]{labelformat = parens, labelsep = space, font = small,labelfont=bf}
\centering
\includegraphics[width=0.90\textwidth]{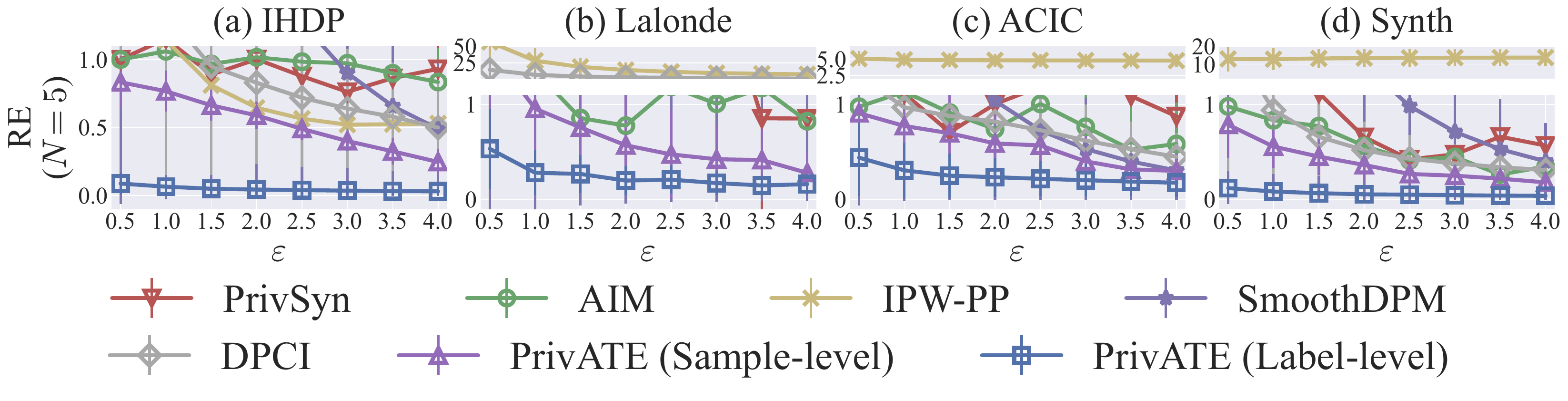}
\vspace{-0.3cm}
\caption{
End-to-end comparison of different methods when the number of matched neighbors $N$ is $5$.
In each plot, the x-axis denotes the privacy budget $\varepsilon$, and the y-axis denotes the relative error.}
\vspace{-0.3cm}
\label{fig:end_to_end}
\end{figure*}

In this section, we perform an end-to-end evaluation of the two levels of \method and the two types of competitors. 
\autoref{fig:end_to_end} illustrates the experimental results on four datasets.

In~\autoref{fig:end_to_end},
different columns represent various datasets.
We have the following observations.
First, as the privacy budget $\varepsilon$ increases,
the 
REs
of all approaches show a downward trend.
The reason is that the increase in the privacy budget reduces the noise intensity,
allowing these methods to capture the feature of the original dataset in a more accurate manner.
On this basis, a lower RE of ATE estimate can be obtained.

Second, the two levels of privacy protection schemes of \method significantly outperform baselines on all datasets.
Label-level privacy of \method performs the best, followed by sample-level privacy of \method.
For the real Lalonde dataset, the REs of most baselines are larger than $1$ even when the privacy budget is $3$, while \method achieves a low RE of less than $0.2$.
This emphasizes the superiority of \method in ATE estimation.
By carefully selecting the matching limit,
\method effectively strikes a balance between the noise perturbation and estimation error, thus obtaining a low ATE estimate error.
For label-level setting of \method, the privacy requirements are lower than those of sample-level,
thus the regression model training and similar sample matching in this setting are better,
making the REs small even when the privacy budget is small.
Sample-level privacy of \method needs to protect all types of variables,
thus the REs are higher than label-level when $\varepsilon$ is low.
As the privacy budget increases to a certain extent (\ie, $\varepsilon=4$), the performance of the two levels is similar.
Moreover, in~\autoref{subsec:appendix_end_to_end}, we find that when the number of matched neighbors $N$ takes different values,
\method shows consistent and similar performance,
which illustrates that \method is robust to the variation of $N$.

\ipw demonstrates the poorest performance across most datasets, particularly under small privacy budgets. 
On the one hand, \ipw requires a fixed threshold to constrain the sample weights in order to achieve bounded sensitivity. However, this approach lacks flexibility and interpretability. 
When the privacy budget is small, the injected noise becomes excessively large, leading to highly inaccurate ATE estimates. 
On the other hand, to ensure privacy protection, \ipw splits the original dataset, using one subset to learn the propensity score function and the other to estimate the ATE. This partitioning further compromises estimation accuracy. 
According to the results in~\autoref{fig:end_to_end}, \ipw fails to effectively handle varying datasets and realistic privacy constraints.

\smdp and \dpci fail to deliver satisfactory performance under small privacy budgets. 
Although these methods generally outperform \ipw,
the estimation accuracy remains limited when strict privacy guarantees are required.
When $\varepsilon \leq 1.5$, the REs of \smdp and \dpci are larger than or close to 1 on most datasets, which indicates poor performance.
This is primarily because they directly inject noise to the final ATE estimate, and smaller privacy budgets lead to higher noise intensity, thereby degrading accuracy. 
As $\varepsilon$ increases, the estimation errors of \dpci and \smdp gradually decrease.
Overall, these two method struggle to achieve low error under strong privacy requirements, which highlights the necessity of~\method.

The relative error of \syn is high, especially when the privacy budget $\varepsilon$ is small.
For instance,
when the privacy budget is $1$,
the REs of \syn exceed $1$ on the four datasets, and the RE on the Lalonde dataset even reaches $4$.
In contrast, 
the RE of \method is less than $1$ in all cases.
This is because \syn is designed to generate a new dataset that closes to the original dataset,
rather than specially designed for ATE estimation.
Due to the small number of samples in the Lalonde dataset and a large income gap between various individuals,
\syn performs poorly in capturing the true data distribution when the privacy budget is low,
resulting in a high RE.
We also notice that for the ACIC dataset,
the REs of \syn are around $1$ at various privacy budgets,
rather than decreasing as the budget increases.
The main reason is that the dimensionality of ACIC is high, which makes it challenging for \syn to capture data characteristics.
In addition, the error variance of \syn under various datasets and privacy settings is significantly higher than that of \method, 
which shows that \syn is not as stable as \method in this task.
 
The performance of~\aim is worse than~\method,
but better than \syn.
On the one hand,
\aim selects the key queries through adaptive and iterative mechanisms,
improving the quality of synthetic data.
On the other hand,
the goal of \syn and \aim is to generate a dataset similar to the original dataset. This type of method aims to achieve promising results on a variety of tasks, but not to achieve SOTA results on a specific task, such as ATE estimation. 
According to the experimental results, \aim and \syn cannot achieve promising performance under low privacy budgets. 
They experience some fluctuation in RE with increasing $\varepsilon$, but overall show a downward trend. 
In contrast, \method is carefully designed to balance noise-induced error and matching error in ATE estimation under DP protection.
Therefore, the performance of \method is better, and the trend in RE becomes more pronounced as $\varepsilon$ increases.

\subsection{Parameter Variation for Label-level Privacy}
\label{section:paramter_study_label_level}

\mypara{Choice of Matching Limit}
Recalling~\autoref{eq:cal_opt_k} and~\autoref{eq:round_clip_opt_k} in~\autoref{section:causal_effect_estimation},
we approximately estimate the total error caused by noise and matching,
and further calculate an optimized matching limit for each sample, which can adaptively vary with the privacy budget and the dataset.

In this section, we verify the rationality of our framework by comparing this adaptive calculation with the fixed value method.
In particular, we only modify the calculation of matching limit $k^*$ in the label-level setting of \method to fixed values (\ie, 1, 10 and 50), and keep the other parts unchanged.
\autoref{fig:compare_fixed_k} illustrates the performance of various matching limit determination mechanisms. 

From~\autoref{fig:compare_fixed_k},
we observe that the fixed value method is difficult to achieve great performance across all datasets.
Since the data characteristics and matching situations of different datasets are various, 
it is not suitable to set the same matching limit on all datasets and privacy budgets.
When the matching limit is set to a small value, a low RE usually be achieved when the privacy budget $\varepsilon$ is small.
As $\varepsilon$ increases, the impact of noise decreases, 
and a small matching limit cannot achieve a great result.
On the other hand, larger $k$ can perform better in high privacy budgets, but this approach will introduce significant errors when the noise intensity is high.
Moreover, for different datasets, the optimal matching limits under the same privacy budget are various, making the determination of the matching upper limit more challenging.

According to the results in~\autoref{fig:compare_fixed_k},
\method almost achieves great performance under various settings.
We find that the RE of the method with fixed small values $k$ does not decrease significantly with the growth of the privacy budget.
The reason is that a small matching limit makes the variation of noise intensity small,
and the matching error does not change under a fixed $k$.
Unlike the above,
as the privacy budget $\varepsilon$ increases,
the matching limit calculated by \method gradually increases,
making the RE decreases.
Even with a small privacy budget,
\method still shows promising performance on these datasets,
reflecting the superiority of adaptive calculation.
For the ACIC dataset, the sample size and dimensionality are both high, which makes the calculated matching limit large.
When $\varepsilon$ is small, the noise error plays a dominant role, making the RE of \method high.
With the increase of $\varepsilon$, \method still shows competitive performance.

\begin{figure*}[!t]
    \centering
    \includegraphics[width=0.90\textwidth]{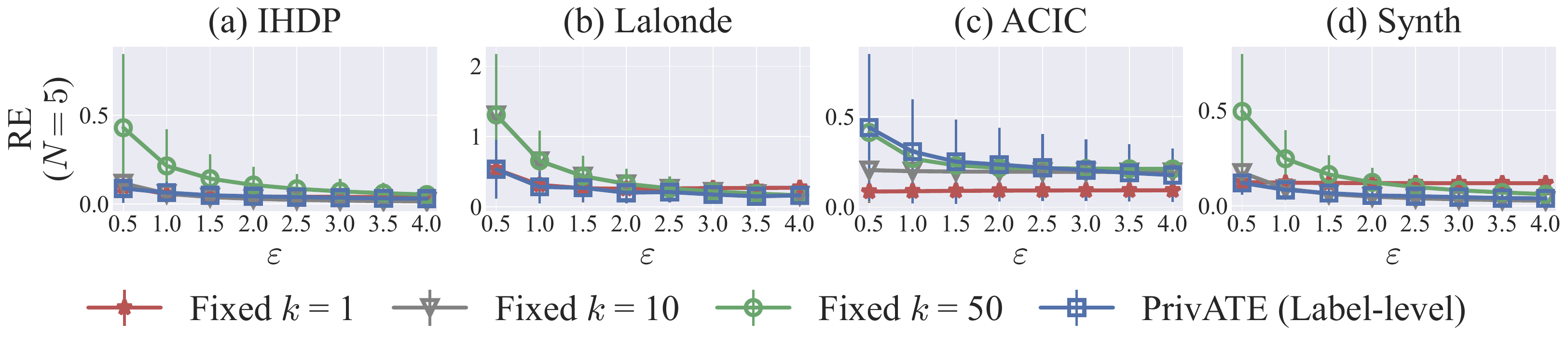}
    \vspace{-0.3cm}
    \caption{
    Impact of different matching limit determination mechanisms in the label-level privacy of \method when the number of matched neighbors $N$ is $5$. 
    The columns represent the used datasets.
    In each plot, the x-axis denotes the privacy budget $\varepsilon$, and the y-axis denotes relative error. 
    }
    \vspace{-0.3cm}
    \label{fig:compare_fixed_k}
\end{figure*}

\mypara{Impact of Error Coefficient in Matching Limit Calculation}
In~\autoref{eq:appro_bias}, 
we utilize an error coefficient $c$ to assist the bias estimate caused by matching in the label-level setting of~\method.
In this section, we explore the impact of various $c$ on final ATE estimation.
A suitable $c$ is crucial to reduce the estimation error.
If the value of $c$ is reasonable, the bias caused by matching can be estimated accurately, which helps to select an ideal matching limit.
However, if $c$ is too large or too small,
the estimation of matching error will be severely distorted, 
resulting in an inappropriate setting of the matching limit and inaccurate estimation of ATE. 

\begin{figure*}[!t]
    \centering
    \includegraphics[width=0.92\textwidth]{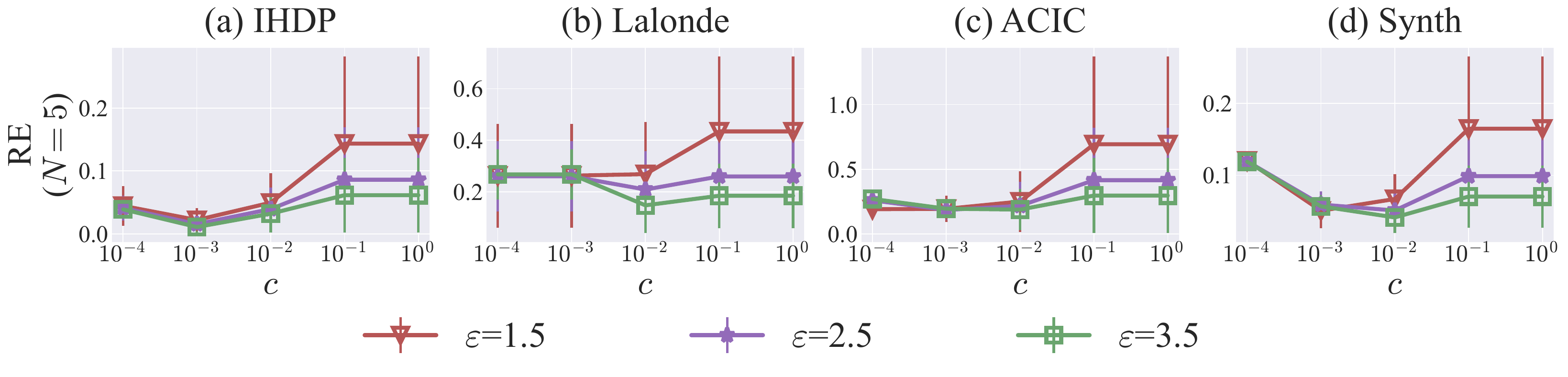}
    \vspace{-0.3cm}
    \caption{
    Impact of different error coefficients $c$ in the label-level privacy 
    when the number of matched neighbors 
    $N$ 
    is $5$. 
    The columns denote
    the used datasets.
    In each plot, 
    the x-axis denotes the error coefficient $c$,
    and the y-axis denotes relative error. 
    }
    \vspace{-0.3cm}
\label{fig:vary_label_c}
\end{figure*}

\autoref{fig:vary_label_c} illustrates the RE of ATE estimate when the number of matched neighbors $N$ is $5$.
We obtain the following observations.
First,
the impact of various $c$ on the final RE is significant.
If $c$ is too small, it means that the bias from matching is overestimated,
which will make the matching limit too low.
If the value of $c$ is too high, the influence of matching will be underestimated,
making the matching limit too large.
Second, for the same dataset,
the optimal value $c$ under various privacy budgets may be different.
The reason is that the noise perturbation and matching error under various privacy budgets are changing.
In addition,
the optimal $c$ for various datasets is also 
different.

In general, we find that \method achieves a great performance under various privacy budgets and datasets when $c=0.01$.
We believe that this setting can accurately characterize the matching error,
thus we adopt $c=0.01$ in our experiments.

\subsection{Parameter Variation for Sample-level Privacy}
\label{section:paramter_study_sample_level}

\mypara{Choice of Matching Limit}
Recalling~\autoref{eq:user_level_opt_k} and~\autoref{eq:user_level_round_clip_opt_k}
in~\autoref{section:causal_effect_estimation},
we calculate the matching limit for sample-level privacy similar to the computation of label-level privacy.
In this section,
we evaluate the performance of the adaptive calculation of \method and fixed matching limit (\ie, 1, 10, and 50) methods, 
which is similar to the comparison in label-level privacy.
\autoref{fig:compare_fixed_k_sample_level}  illustrates the related results.

We observe that the relative error for a large fixed matching limit is significantly higher than that for a small fixed matching limit.
The reason is that sample-level privacy allocates the total privacy budget across multiple phases to satisfy DP. 
As a result, the privacy budget available in the final phase is smaller. 
In addition, due to noise perturbation in both the regression model and the propensity score, the matching results are not entirely accurate. 
Consequently, increasing the matching limit has less impact compared to label-level privacy.

At the same time, setting the matching limit to a very small fixed value is not optimal. 
When noise has low interference in matching or the privacy budget is large enough, moderately increasing the matching limit can help reduce estimation error. 
According to the results in~\autoref{fig:compare_fixed_k_sample_level}, our method achieves promising performance in various settings.

\begin{figure*}[!t]
    \centering
    \includegraphics[width=0.90\textwidth]{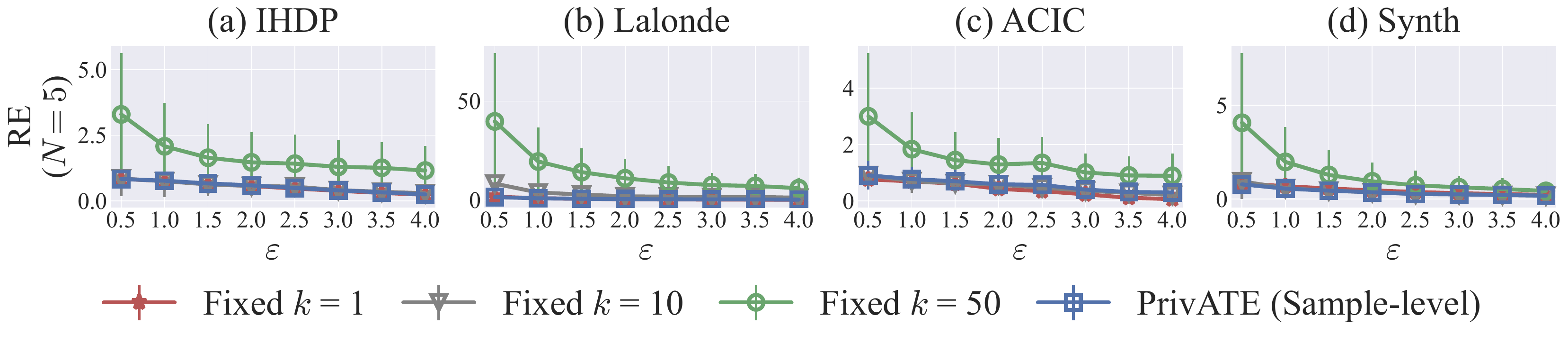}
    \vspace{-0.3cm}
    \caption{
    Impact of different matching limit determination mechanisms in the sample-level privacy of \method when the number of matched neighbors $N$ is $5$. 
    The columns represent the used datasets.
    In each plot, the x-axis denotes the privacy budget $\varepsilon$, and the y-axis denotes relative error. 
    }
    \vspace{-0.3cm}
    \label{fig:compare_fixed_k_sample_level}
\end{figure*}

\mypara{Impact of Error Coefficient in Matching Limit Calculation}
Recalling~\autoref{eq:user_level_opt_k} in~\autoref{section:causal_effect_estimation},
an error coefficient $h$ is utilized to help determine the value of matching limit in the sample-level privacy of \method.
In the section, we explore the influence of different $h$ on the ATE estimate.
\autoref{fig:vary_user_h} illustrates the relative errors of various $h$.

We observe that a larger $h$ tends to produce higher relative errors.
The reason is that the noise is injected into each phase in the sample-level setting,
thus the fidelity of regression model and matching results is significantly lower than that of the label-level.
In this case,
increasing $h$ cannot effectively reduce the matching error.
On the other hand,
the sensitivity will increase as $h$ grows,
making the noise error higher.
At the same time, a relatively small value of $h$ may not be a great choice since this cannot reduce the matching error.
Furthermore, it is impossible to find an $h$ that achieves the lowest RE for various privacy budgets.
Taking into account the impact of different budgets and datasets,
we choose to set $h=0.001$ in the experiments.

\begin{figure*}[!t]
    \centering
    \includegraphics[width=0.90\textwidth]{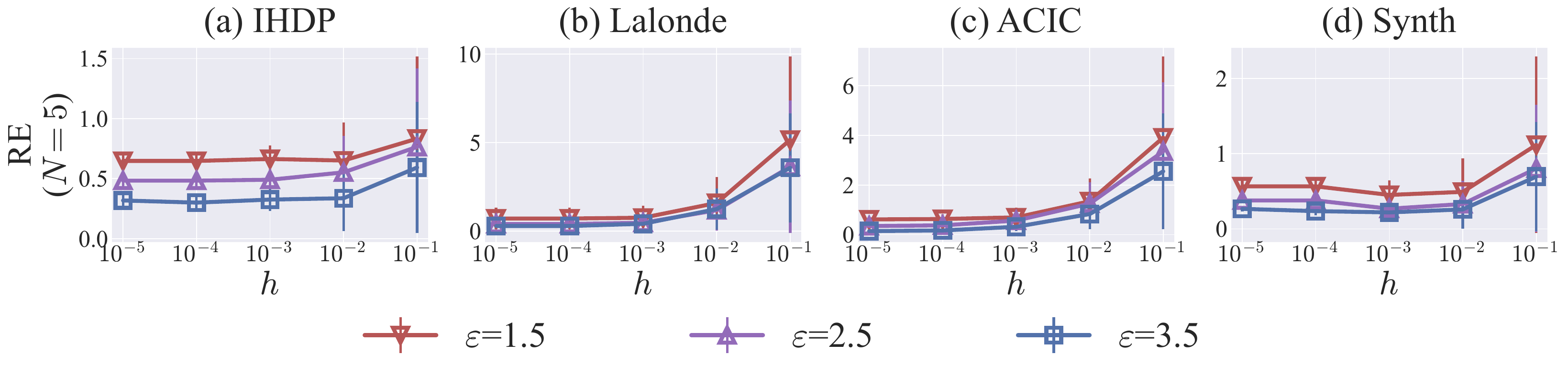}
    \caption{
    Impact of different error coefficients $h$ in the sample-level privacy 
    when the number of matched neighbors $N$ is $5$. 
    The columns 
    denote
    the used datasets.
    In each plot, 
    the x-axis denotes the error coefficient $h$,
    and the y-axis denotes relative error. 
    }
    \label{fig:vary_user_h}
\end{figure*}

\section{Discussion}
\label{sec:discussion}

\mypara{Generalization}
In this paper, our proposed \method mainly focuses on the binary treatments.
For more complex causal inference setups (\eg, multi-valued treatments),
the potential schemes are as follows:
A multinomial logit regression model can be applied for fitting. Then, the noise can be injected into the true probability vector. 
The multi-dimensional treatments can be perturbed by Generalized Randomized Response (GRR). 
Moreover, a distance metric function can be utilized to calculate the distance between different categories and achieve pairwise privacy-preserving ATE estimation.

\mypara{Scalability}
For \method, propensity score matching and adaptive matching algorithms involve non-trivial computation,
which may be computationally slow on large-scale datasets.
Here, we clarify that \method achieves a low computational complexity compared to existing solutions according to the complexity analysis in~\autoref{app:complexity_analysis}.
This suggests that \method has the potential to run on larger datasets.
Furthermore, to address the scalability bottlenecks that matching algorithms may encounter on large datasets, the efficiency of \method can be further improved by employing approaches such as approximate nearest neighbor (ANN) matching, hierarchical reduction matching, and GPU acceleration.

\section{Related Work}
\label{sec:related_work}

There are several literatures explore differentially private ATE estimation in observational studies~\cite{lee2019privacy,imbens2015causal,guha2025differentially,schroder2025private,ohnishi2024differentially,koga2024differentially,lebeda2025model}.
Specifically,
Guha~\etal~\cite{guha2025differentially} design a differentially private weighted average treatment effect estimator for binary outcomes by splitting the data into several disjoint groups.
Similarly, Lebeda~\etal~\cite{lebeda2025model} also propose a data splitting-based framework to estimate the average treatment effect.
Ohnishi~\etal~\cite{ohnishi2024differentially}
present a differentially private covariate balancing weighting estimator to infer causal effects while protecting the privacy of covariates.
Schr{\"o}der~\etal~\cite{schroder2025private} further propose a framework to estimate the ATE by doubly robust estimation and construct the confidence intervals.
In this setting, it is challenging to achieve promising performance under small privacy budgets, since the noise is directly injected to the ATE estimate.
In addition,
Lee~\etal~\cite{lee2019privacy} propose a privacy-preserving inverse probability weighting (IPW) method~\cite{imbens2015causal} to estimate the causal effect.
However, this approach relies on a pre-defined truncation threshold to bound the sample weights, which lacks flexibility and interpretability.
Koga~\etal~\cite{koga2024differentially} introduce a smooth-sensitivity-based DP algorithm to perturb the true average treatment effect. 
Nevertheless, it requires that the matching variables are discrete, and performs poorly under strong privacy constraints.

Moreover, some research incorporates differential privacy to protect real data in randomized experiments~\cite{kancharla2021robust,betlei2021differentially,ohnishi2025locally,farzam2024causal,mukherjee2024improving,kusner2016private}.
Kancharla~\etal~\cite{kancharla2021robust} investigate the problem of ATE estimation in randomized controlled trials. 
They assume a binary outcome space and propose two consistent estimators for estimating the ATE.
Betlei~\etal~\cite{betlei2021differentially} focus on privacy-preserving individual treatment effect (ITE) estimation and introduce a differentially private method, ADUM, which learns uplift models from data aggregated according to a given partition of the feature space.
Javanmard~\etal~\cite{javanmard2023causal} propose a differential privacy mechanism, CLUSTER-DP, which leverages the inherent cluster structure of the data to estimate causal effects, while perturbing the outcomes to preserve individual privacy.
Furthermore, Ohnishi~\etal~\cite{ohnishi2025locally} develop a method for inferring causal effects from locally privatized data in randomized experiments.
In addition,
Niu~\etal~\cite{niu2022differentially} introduce a meta-algorithm for estimating conditional average treatment effects using DP-EBMs~\cite{nori2021accuracy} as the base learner.
Schr{\"o}der~\etal~\cite{schroder2025differentially} further propose a framework for conditional average treatment effects estimation that is Neyman-orthogonal.

In addition, differentially private data synthesis can also be used for differentially private ATE estimation~\cite{zhang2021privsyn,mckenna2022aim,mckenna2019graphical,zhang2017privbayes,vietri2020new,cai2021data}.
In this way,
the ATE estimate can be calculated based on a synthetic dataset that satisfies DP.
Zhang~\etal~\cite{zhang2021privsyn} design a new method to automatically and privately identify correlations in the data,
and then generate sample data from a dense graphic model.
McKenna~\etal~\cite{mckenna2022aim} propose a workload-adaptive algorithm
that first selects a set of queries, then privately measures those queries, and finally generates synthetic data from the noisy measurements.
However, these approaches are essentially different from our work:
Their goal is to generate a synthetic dataset that closely resembles the original one under DP~\cite{yuan2023privgraph,yuan2025privload,DHZFCZG23,wang2023privtrace,yuan2024psgraph},
while our focus is on accurately estimating ATE 
while satisfying DP.
The experimental results also demonstrate the superiority of our method.

Moreover, there are also some other privacy-preserving solutions (\eg, $k$-anonymity, secure multi-party computation) that can be used for data protection.
For instance,
Abadi~\etal~\cite{abadi2016deep} propose DP-SGD, which is a classic method to train a model while ensuring the privacy of training samples.
Davidson~\etal~\cite{davidson2022star} design a practical mechanism named STAR for providing cryptographically-enforced $k$-anonymity protections.
Furthermore,
Shamsabadi~\etal~\cite{shamsabadi2025nebula} present Nebula, a system for differentially private histogram estimation on data distributed among clients.
Although these methods cannot be directly applied to differentially private ATE estimation, their underlying ideas offer valuable insights for future research.

\section{Conclusion}
\label{sec:conclusion}

In this paper, we propose a practical framework~\method for estimating the average treatment effect (ATE) for observational data under differential privacy (DP).
Based on propensity score matching,
two different levels (\ie, label-level and sample-level) of privacy protection approaches in~\method are proposed to accommodate various privacy requirements.
To strike a great trade-off between noise and matching errors,
\method achieves an adaptive matching limit determination by considering the joint influence caused by noise perturbation and matching inaccuracy.
Extensive experiments on four datasets demonstrate the superiority of our proposed~\method.
We further verify the effectiveness of matching limit determination.
We also analyze the impact of hyper-parameters of~\method and provide the guideline for their selection.

\section*{Ethics Considerations}
This paper focuses on differentially private average treatment effect estimation.
We strictly followed ethical guidelines by using publicly available, open-source datasets, under licenses that permit research and educational use.
As these datasets were curated and released by third parties, direct informed consent was not applicable. 
However, we are committed to ethical data use and will
comply with all licensing terms for any future modifications
or redistribution.

\section*{Acknowledgment}
We would like to thank the anonymous reviewers for their insightful comments. 
This work is supported 
in part 
by the National Natural Science Foundation of China under Grants No. (62402431, 62441618, 62025206, U23A20296, U24A20237, 62402379, 72594583011, 7257010373),
Key R\&D Program of Zhejiang Province under Grants No. (2024C01259, 2025C01061, 2024C01065, 2024C01012, 2025C01089),
the project CiCS of the research programme Gravitation which is (partly) financed by the Dutch Research Council (NWO) under Grant 024.006.037,
the China Postdoctoral Science Foundation under Grants No. (2025M771501, BX20250380),
and Zhejiang University.

{
    \footnotesize
    \bibliography{easy}

\begin{thebibliography}{10}
\providecommand{\url}[1]{#1}
\csname url@samestyle\endcsname
\providecommand{\newblock}{\relax}
\providecommand{\bibinfo}[2]{#2}
\providecommand{\BIBentrySTDinterwordspacing}{\spaceskip=0pt\relax}
\providecommand{\BIBentryALTinterwordstretchfactor}{4}
\providecommand{\BIBentryALTinterwordspacing}{\spaceskip=\fontdimen2\font plus
\BIBentryALTinterwordstretchfactor\fontdimen3\font minus \fontdimen4\font\relax}
\providecommand{\BIBforeignlanguage}[2]{{%
\expandafter\ifx\csname l@#1\endcsname\relax
\typeout{** WARNING: IEEEtran.bst: No hyphenation pattern has been}%
\typeout{** loaded for the language `#1'. Using the pattern for}%
\typeout{** the default language instead.}%
\else
\language=\csname l@#1\endcsname
\fi
#2}}
\providecommand{\BIBdecl}{\relax}
\BIBdecl

\bibitem{brand2023recent}
J.~E. Brand, X.~Zhou, and Y.~Xie, ``{Recent Developments in Causal Inference and Machine Learning},'' \emph{Annual Review of Sociology}, vol.~49, no.~1, pp. 81--110, 2023.

\bibitem{glass2013causal}
T.~A. Glass, S.~N. Goodman, M.~A. Hern{\'a}n, and J.~M. Samet, ``{Causal Inference in Public Health},'' \emph{Annual Review of Public Health}, vol.~34, no.~1, pp. 61--75, 2013.

\bibitem{varian2016causal}
H.~R. Varian, ``{Causal Inference in Economics and Marketing},'' \emph{Proceedings of the National Academy of Sciences}, vol. 113, no.~27, pp. 7310--7315, 2016.

\bibitem{pearl2003statistics}
J.~Pearl, ``{Statistics and Causal Inference: A Review},'' \emph{Test}, vol.~12, pp. 281--345, 2003.

\bibitem{matthay2022causal}
E.~C. Matthay and M.~M. Glymour, ``{Causal Inference Challenges and New Directions for Epidemiologic Research on the Health Effects of Social Policies},'' \emph{Current Epidemiology Reports}, vol.~9, no.~1, pp. 22--37, 2022.

\bibitem{cordero2018causal}
J.~M. Cordero, V.~Crist{\'o}bal, and D.~Sant{\'\i}n, ``{Causal Inference on Education Policies: A Survey of Empirical Studies using PISA, TIMSS and PIRLS},'' \emph{Journal of Economic Surveys}, 2018.

\bibitem{shalit2017estimating}
U.~Shalit, F.~D. Johansson, and D.~Sontag, ``{Estimating Individual Treatment Effect: Generalization Bounds and Algorithms},'' in \emph{International Conference on Machine Learning}.\hskip 1em plus 0.5em minus 0.4em\relax PMLR, 2017, pp. 3076--3085.

\bibitem{yao2021survey}
L.~Yao, Z.~Chu, S.~Li, Y.~Li, J.~Gao, and A.~Zhang, ``{A Survey on Causal Inference},'' \emph{ACM Transactions on Knowledge Discovery from Data (TKDD)}, vol.~15, no.~5, pp. 1--46, 2021.

\bibitem{berrevoets2023impute}
J.~Berrevoets, F.~Imrie, T.~Kyono, J.~Jordon, and M.~Van~der Schaar, ``{To Impute or Not to Impute? Missing Data in Treatment Effect Estimation},'' in \emph{International Conference on Artificial Intelligence and Statistics}.\hskip 1em plus 0.5em minus 0.4em\relax PMLR, 2023, pp. 3568--3590.

\bibitem{rosenbaum1983central}
P.~R. Rosenbaum and D.~B. Rubin, ``{The Central Role of The Propensity Score in Observational Studies for Causal Effects},'' \emph{Biometrika}, vol.~70, no.~1, pp. 41--55, 1983.

\bibitem{kusner2016private}
M.~J. Kusner, Y.~Sun, K.~Sridharan, and K.~Q. Weinberger, ``{Private Causal Inference},'' in \emph{Artificial Intelligence and Statistics}.\hskip 1em plus 0.5em minus 0.4em\relax PMLR, 2016, pp. 1308--1317.

\bibitem{dwork2006calibrating}
C.~Dwork, F.~McSherry, K.~Nissim, and A.~Smith, ``{Calibrating Noise to Sensitivity in Private Data Analysis},'' in \emph{Theory of Cryptography Conference}.\hskip 1em plus 0.5em minus 0.4em\relax Springer, 2006, pp. 265--284.

\bibitem{bittau2017prochlo}
A.~Bittau, {\'U}.~Erlingsson, P.~Maniatis, I.~Mironov, A.~Raghunathan, D.~Lie, M.~Rudominer, U.~Kode, J.~Tinnes, and B.~Seefeld, ``{Prochlo: Strong Privacy for Analytics in the Crowd},'' in \emph{SOSP}, 2017.

\bibitem{rogers2020linkedin}
R.~Rogers, S.~Subramaniam, S.~Peng, D.~Durfee, S.~Lee, S.~K. Kancha, S.~Sahay, and P.~Ahammad, ``{Linkedin's Audience Engagements Api: A Privacy Preserving Data Analytics System at Scale},'' \emph{CoRR abs/2002.05839}, 2020.

\bibitem{yuan2024privcpm}
Q.~Yuan, M.~Sun, Y.~Sheng, and Q.~Guo, ``{PrivCPM: Privacy-Preserving Cooperative Pricing Mechanism in Coupled Power-Traffic Networks},'' \emph{IEEE Transactions on Smart Grid}, vol.~16, no.~1, pp. 612--626, 2025.

\bibitem{wang2021continuous}
T.~Wang, J.~Q. Chen, Z.~Zhang, D.~Su, Y.~Cheng, Z.~Li, N.~Li, and S.~Jha, ``{Continuous Release of Data Streams under both Centralized and Local Differential Privacy},'' in \emph{ACM CCS}, 2021, pp. 1237--1253.

\bibitem{lee2019privacy}
S.~K. Lee, L.~Gresele, M.~Park, and K.~Muandet, ``{Privacy-preserving Causal Inference via Inverse Probability Weighting},'' \emph{CoRR abs/1905.12592}, 2019.

\bibitem{imbens2015causal}
G.~W. Imbens and D.~B. Rubin, \emph{{Causal Inference in Statistics, Social, and Biomedical Sciences}}.\hskip 1em plus 0.5em minus 0.4em\relax Cambridge University Press, 2015.

\bibitem{guha2025differentially}
S.~Guha and J.~P. Reiter, ``{Differentially Private Estimation of Weighted Average Treatment Effects for Binary Outcomes},'' \emph{Computational Statistics \& Data Analysis}, p. 108145, 2025.

\bibitem{ohnishi2024differentially}
Y.~Ohnishi and J.~Awan, ``{Differentially Private Covariate Balancing Causal Inference},'' \emph{CoRR abs/2410.14789}, 2024.

\bibitem{koga2024differentially}
T.~Koga, K.~Chaudhuri, and D.~Page, ``{Differentially Private Multi-Site Treatment Effect Estimation},'' in \emph{2024 IEEE Conference on Secure and Trustworthy Machine Learning (SaTML)}.\hskip 1em plus 0.5em minus 0.4em\relax IEEE, 2024, pp. 472--489.

\bibitem{rosenbaum1987model}
P.~R. Rosenbaum, ``{Model-based Direct Adjustment},'' \emph{Journal of the American statistical Association}, vol.~82, no. 398, pp. 387--394, 1987.

\bibitem{dwork2014algorithmic}
C.~Dwork, A.~Roth \emph{et~al.}, ``{The Algorithmic Foundations of Differential Privacy},'' \emph{Foundations and Trends{\textregistered} in Theoretical Computer Science}, vol.~9, no. 3--4, pp. 211--407, 2014.

\bibitem{wang2016using}
Y.~Wang, X.~Wu, and D.~Hu, ``{Using Randomized Response for Differential Privacy Preserving Data Collection},'' in \emph{EDBT/ICDT Workshops}, vol. 1558, 2016, pp. 0090--6778.

\bibitem{zhang2018calm}
Z.~Zhang, T.~Wang, N.~Li, S.~He, and J.~Chen, ``{CALM: Consistent Adaptive Local Marginal for Marginal Release under Local Differential Privacy},'' in \emph{ACM CCS}, 2018, pp. 212--229.

\bibitem{du2021ahead}
L.~Du, Z.~Zhang, S.~Bai, C.~Liu, S.~Ji, P.~Cheng, and J.~Chen, ``{AHEAD: Adaptive Hierarchical Decomposition for Range Query under Local Differential Privacy},'' in \emph{ACM CCS}, 2021, pp. 1266--1288.

\bibitem{wei2023dpmlbench}
C.~Wei, M.~Zhao, Z.~Zhang, M.~Chen, W.~Meng, B.~Liu, Y.~Fan, and W.~Chen, ``{DPMLBench: Holistic Evaluation of Differentially Private Machine Learning},'' in \emph{ACM CCS}, 2023, pp. 2621--2635.

\bibitem{mckenna2022aim}
R.~McKenna, B.~Mullins, D.~Sheldon, and G.~Miklau, ``{AIM: An Adaptive and Iterative Mechanism for Differentially Private Synthetic Data},'' \emph{Proceedings of the VLDB Endowment}, vol.~15, no.~11, pp. 2599--2612, 2022.

\bibitem{zhang2021privsyn}
Z.~Zhang, T.~Wang, J.~Honorio, N.~Li, M.~Backes, S.~He, J.~Chen, and Y.~Zhang, ``{PrivSyn: Differentially Private Data Synthesis},'' in \emph{USENIX Security Symposium}, 2021.

\bibitem{schroder2025private}
M.~Schr{\"o}der, J.~Hartenstein, and S.~Feuerriegel, ``{PrivATE: Differentially Private Confidence Intervals for Average Treatment Effects},'' \emph{CoRR abs/2505.21641}, 2025.

\bibitem{hill2011bayesian}
J.~L. Hill, ``{Bayesian Nonparametric Modeling for Causal Inference},'' \emph{Journal of Computational and Graphical Statistics}, vol.~20, no.~1, pp. 217--240, 2011.

\bibitem{dehejia1999causal}
R.~H. Dehejia and S.~Wahba, ``{Causal Effects in Nonexperimental Studies: Reevaluating the Evaluation of Training Programs},'' \emph{Journal of the American Statistical Association}, 1999.

\bibitem{dorie2019automated}
V.~Dorie, J.~Hill, U.~Shalit, M.~Scott, and D.~Cervone, ``{Automated Versus Do-it-yourself Methods for Causal Inference: Lessons Learned from a Data Analysis Competition},'' \emph{Statistical Science}, 2019.

\bibitem{lebeda2025model}
C.~Lebeda, M.~Even, A.~Bellet, and J.~Josse, ``{Model Agnostic Differentially Private Causal Inference},'' \emph{CoRR abs/2505.19589}, 2025.

\bibitem{kancharla2021robust}
M.~Kancharla and H.~Kang, ``{A Robust, Differentially Private Randomized Experiment for Evaluating Online Educational Programs with Sensitive Student Data},'' \emph{CoRR abs/2112.02452}, 2021.

\bibitem{betlei2021differentially}
A.~Betlei, T.~Gregoir, T.~Rahier, A.~Bissuel, E.~Diemert, and M.-R. Amini, ``{Differentially Private Individual Treatment Effect Estimation from Aggregated Data},'' in \emph{FOCS Workshop}, 2021.

\bibitem{ohnishi2025locally}
Y.~Ohnishi and J.~Awan, ``{Locally Private Causal Inference for Randomized Experiments},'' \emph{Journal of Machine Learning Research}, vol.~26, no.~14, pp. 1--40, 2025.

\bibitem{farzam2024causal}
A.~Farzam and G.~Sapiro, ``{Causal Inference under Differential Privacy: Challenges and Mitigation Strategies},'' in \emph{NeurIPS 2024 Causal Representation Learning Workshop}, 2024.

\bibitem{mukherjee2024improving}
S.~Mukherjee, A.~Mustafi, A.~Slavkovic, and L.~Vilhuber, ``{Improving Privacy for Respondents in Randomized Controlled Trials: A Differential Privacy Approach},'' National Bureau of Economic Research, Tech. Rep., 2024.

\bibitem{javanmard2023causal}
A.~Javanmard, V.~Mirrokni, and J.~Pouget-Abadie, ``{Causal Inference with Differentially Private (Clustered) Outcomes},'' \emph{CoRR abs/2308.00957}, 2023.

\bibitem{niu2022differentially}
F.~Niu, H.~Nori, B.~Quistorff, R.~Caruana, D.~Ngwe, and A.~Kannan, ``{Differentially Private Estimation of Heterogeneous Causal Effects},'' in \emph{Conference on Causal Learning and Reasoning}, 2022, pp. 618--633.

\bibitem{nori2021accuracy}
H.~Nori, R.~Caruana, Z.~Bu, J.~H. Shen, and J.~Kulkarni, ``{Accuracy, Interpretability, and Differential Privacy via Explainable Boosting},'' in \emph{International Conference on Machine Learning}, 2021, pp. 8227--8237.

\bibitem{schroder2025differentially}
M.~Schr{\"o}der, V.~Melnychuk, and S.~Feuerriegel, ``{Differentially Private Learners for Heterogeneous Treatment Effects},'' in \emph{ICLR}, 2025.

\bibitem{mckenna2019graphical}
R.~McKenna, D.~Sheldon, and G.~Miklau, ``{Graphical-Model Based Estimation and Inference for Differential Privacy},'' in \emph{International Conference on Machine Learning}.\hskip 1em plus 0.5em minus 0.4em\relax PMLR, 2019, pp. 4435--4444.

\bibitem{zhang2017privbayes}
J.~Zhang, G.~Cormode, C.~M. Procopiuc, D.~Srivastava, and X.~Xiao, ``{PrivBayes: Private Data Release via Bayesian Networks},'' \emph{ACM Transactions on Database Systems (TODS)}, vol.~42, no.~4, pp. 1--41, 2017.

\bibitem{vietri2020new}
G.~Vietri, G.~Tian, M.~Bun, T.~Steinke, and S.~Wu, ``{New Oracle-Efficient Algorithms for Private Synthetic Data Release},'' in \emph{International Conference on Machine Learning}.\hskip 1em plus 0.5em minus 0.4em\relax PMLR, 2020, pp. 9765--9774.

\bibitem{cai2021data}
K.~Cai, X.~Lei, J.~Wei, and X.~Xiao, ``{Data Synthesis via Differentially Private Markov Random Fields},'' \emph{Proceedings of the VLDB Endowment}, vol.~14, no.~11, pp. 2190--2202, 2021.

\bibitem{yuan2023privgraph}
Q.~Yuan, Z.~Zhang, L.~Du, M.~Chen, P.~Cheng, and M.~Sun, ``{PrivGraph: Differentially Private Graph Data Publication by Exploiting Community Information},'' in \emph{USENIX Security Symposium}, 2023.

\bibitem{yuan2025privload}
Q.~Yuan, H.~Wu, S.~He, and M.~Sun, ``{PrivLoad: Privacy-preserving Load Profiles Synthesis Based on Diffusion Models},'' \emph{IEEE Transactions on Smart Grid}, vol.~16, no.~6, pp. 5628--5640, 2025.

\bibitem{DHZFCZG23}
Y.~Du, Y.~Hu, Z.~Zhang, Z.~Fang, L.~Chen, B.~Zheng, and Y.~Gao, ``{LDPTrace: Locally Differentially Private Trajectory Synthesis},'' \emph{{Proceedings of the VLDB Endowment}}, 2023.

\bibitem{wang2023privtrace}
H.~Wang, Z.~Zhang, T.~Wang, S.~He, M.~Backes, J.~Chen, and Y.~Zhang, ``{PrivTrace: Differentially Private Trajectory Synthesis by Adaptive Markov Model},'' in \emph{USENIX Security Symposium}, 2023.

\bibitem{yuan2024psgraph}
Q.~Yuan, Z.~Zhang, L.~Du, M.~Chen, M.~Sun, Y.~Gao, M.~Backes, S.~He, and J.~Chen, ``{PSGraph: Differentially Private Streaming Graph Synthesis by Considering Temporal Dynamics},'' \emph{CoRR abs/2412.11369}, 2024.

\bibitem{abadi2016deep}
M.~Abadi, A.~Chu, I.~Goodfellow, H.~B. McMahan, I.~Mironov, K.~Talwar, and L.~Zhang, ``{Deep Learning with Differential Privacy},'' in \emph{ACM CCS}, 2016, pp. 308--318.

\bibitem{davidson2022star}
A.~Davidson, P.~Snyder, E.~Quirk, J.~Genereux, B.~Livshits, and H.~Haddadi, ``{STAR: Secret Sharing for Private Threshold Aggregation Reporting},'' in \emph{ACM CCS}, 2022, pp. 697--710.

\bibitem{shamsabadi2025nebula}
A.~S. Shamsabadi, P.~Snyder, R.~Giles, A.~Bellet, and H.~Haddadi, ``{Nebula: Efficient, Private and Accurate Histogram Estimation},'' in \emph{ACM CCS}, 2025.

\end{thebibliography}
    \bibliographystyle{IEEEtran}
}

\appendix

\subsection{Workflow of~\method}
\label{sec:appendix_overall_workflow}

\autoref{algorithm:put_together} illustrates the overall process of~\method.
In the label-level privacy,
only the third phase needs to consume privacy budget to perturb the outcomes.
While in the sample-level privacy,
all three phases need to allocate the privacy budget to protect the privacy of all types of data.

\begin{algorithm}[!htbp]

        \caption{\method}
        
        \label{algorithm:put_together}
        \KwIn{Original dataset $D$,
        privacy level $l$,
        privacy budget $\varepsilon$ (label-level) or $\varepsilon=\varepsilon_{11}+\varepsilon_{12}+\varepsilon_2+\varepsilon_3$ (sample-level), the number of matched neighbors $N$, the maximum variation range of outcome $B$,  the error coefficient $c$ (label-level) or $h$ (sample-level)
        }
        \KwOut{
        Propensity score $e'(X)$, treatment $T'$, sorted matrices $H$,
        average treatment effect estimate $\hat\tau$}
        \If{l is label-level}
        {
        \Comment{Phase 1: Regression model training}
        $e'(X)\leftarrow$ \autoref{algorithm:phase_1}($D,l$)    \\
        \Comment{Phase 2: Similar sample matching}
        $T',H\leftarrow$\autoref{algorithm:phase_2}$(D,e'(X),l)$  \\
        \Comment{Phase 3: Causal effect estimation}
        $\hat \tau\leftarrow$\autoref{algorithm:phase_3}$(D,H,N,T',l,B,\varepsilon,c)$
        }
        \Else
        {
        \Comment{Phase 1: Regression model training}
        $e'(X)\leftarrow$ \autoref{algorithm:phase_1}($D,l,\varepsilon_{11},\varepsilon_{12}$) \\
        \Comment{Phase 2: Similar sample matching}
        $T',H\leftarrow$\autoref{algorithm:phase_2}$(D,e'(X),l,\varepsilon_2)$ \\
        \Comment{Phase 3: Causal effect estimation}
        $\hat \tau\leftarrow$\autoref{algorithm:phase_3}$(D,H,N,T',l,B,\varepsilon_3,h)$
        }
        
\end{algorithm}

\section{Detailed Proofs}
\label{sec:proof_appendix}

\subsection{Proof of \autoref{throrem:label_level}}

\begin{proof}
\label{proof:label_level_DP}

In the label-level setting, only the outcome is sensitive information.
    In the first two phases, the outcome does not need to be accessed, thus does not consume the privacy budget.
    In the third phase, the original potential outcomes of all samples are aggregated to calculate the potential outcome sum of the treated and control groups.
    We utilize Laplace mechanism to perturb the original aggregated outcome, thus this step satisfies $\varepsilon$-DP.
    Note that the samples of the treated
     and control groups are non-overlapping, thus these
    two parts can share the same privacy budget according to the parallel composition.
    Finally, the final ATE estimate is computed based on the perturbed aggregated outcome, which can be considered post-processing and does not consume the privacy budget.
    Therefore, if the privacy level $l$ is label-level,
    \autoref{algorithm:put_together} satisfies $\varepsilon$-Label~DP.

\end{proof}

\begin{lemma}
\label{lemma}
    Let $G(w)$ and $g(w)$ be two vector-based functions, which are continuous, and differentiable at all points. Moreover, let $G(w)$ and $G(w)+g(w)$ be $\lambda_1$-strongly convex in $L_1$-norm. 
    If $w_1=\arg\min_w G(w)$ and $w_2=\arg\min_w G(w)+g(w)$, then 
    \begin{equation*}
        ||w_1-w_2||_1 \leq \frac{1}{\lambda_1}\max_w||\nabla g(w)||_{\infty}
    \end{equation*}
    
\end{lemma}

\begin{proof}
    Using the definition of $w_1$ and $w_2$, and the fact that $G$ and $g$ are continuous and differentiable everywhere,
    \begin{equation*}
        \nabla G(w_1)-\nabla G(w_2)=\nabla g(w_2)
    \end{equation*}
    
    As $G(w)$ is $\lambda_1$-strongly convex, then
    \begin{align*}
        (\nabla G(w_1)-\nabla G(w_2))^T(w_1-w_2) &= (\nabla g(w_2))^T(w_1-w_2) \\ &\geq \lambda_1 ||w_1-w_2||_1^2
    \end{align*}

    Based on Hölder inequality, we have
    \begin{equation*}
        (\nabla g(w_2))^T(w_1-w_2) \leq ||\nabla g(w_2)||_{\infty}||w_1-w_2||_1
    \end{equation*}

    Then, we obtain 
    \begin{equation*}
        \lambda_1 ||w_1-w_2||_1^2 \leq ||\nabla g(w_2)||_{\infty}||w_1-w_2||_1
    \end{equation*}

    Further, we have
    \begin{equation*}
        \lambda_1 ||w_1-w_2||_1 \leq ||\nabla g(w_2)||_{\infty}
    \end{equation*}

    Finally, we can obtain
    \begin{equation*}
         ||w_1-w_2||_1 \leq\frac{1}{\lambda_1} \max_{w} ||\nabla g(w)||_{\infty}
    \end{equation*}
\end{proof}

\subsection{Proof of \autoref{throrem:user_level}}
\label{proof:user_level_DP}
\mypara{Proof 1: Private model training satisfies $\varepsilon_{11}$-DP}

\begin{proof}
Assuming that the dataset $D$ and the adjacent dataset $D'$ differ in the $i$-th sample, the difference between their corresponding loss functions is as follows:

\begin{align*}
    g(w) &= J(w,D)-J(w,D') \\
    &=\frac{1}{n}[ \log(1+e^{-{X}_i^T w t_i})-\log(1+e^{-{X'}_i^T w t'_i})]
\end{align*}

We observe that the logarithmic loss function $l=\log(1+e^{-X^Twt})$ is convex and differentiable, and $J(w)$ is $\lambda$-strongly convex. 
Then, we can obtain the derivative of $g(w)$ as follows:

\begin{equation*}
    \nabla g(w)= \frac{1}{n}[X_i^T l'(X_i^Twt_i)t_i-{X'}_i^T l'({X'}_i^T w t'_i) t'_i],
\end{equation*}
where $l'(\cdot)$ stands for the derivative of the logarithmic loss function, and its range is $[0,1]$.
Then, we obtain
\begin{equation*}
    ||\nabla g(w)||_{\infty}= \frac{1}{n}||X_i^T l'(X_i^Twt_i)t_i-{X'}_i^T l'({X'}_i^T w t'_i) t'_i||_{\infty}
\end{equation*}

For vector $v$, we have $||v||_{\infty}=\max_{j}{|v_j|}$.
Since $|l'(\cdot)t|\leq1$,
we further obtain
\begin{equation*}
    |\nabla g(w)_j|\leq\frac{1}{n}|X_{i,j}-X'_{i,j}|\leq\frac{1}{n}(|X_{i,j}|+|X'_{i,j}|)\leq\frac{2}{n}
\end{equation*}

Therefore, we can obtain
\begin{equation*}
    \max_{w}||\nabla g(w)||_{\infty}\leq \frac{2}{n}
\end{equation*}

It is obvious that $\frac{\lambda}{2}||w||^2_2$ is $\lambda$-convex in $L_2$-norm.
Next, we need to convert it to $L_1$-norm strongly convex.
For vector $v$, $||v||_1\leq\sqrt{d}||v||_2$,
then 
\begin{equation*}
     \lambda||w_1-w_2||_2^2 \geq \lambda\frac{1}{d}||w_1-w_2||_1^2
\end{equation*}

Therefore, the regularization term $\frac{\lambda}{2}||w||^2$ is $\frac{\lambda}{d}$-strongly convex in $L_1$-norm (\ie, $\lambda_1=\frac{\lambda}{d}$).
Based on Lemma~\ref{lemma}, we can obtain

\begin{equation*}
    ||w_1-w_2||_1 \leq \frac{1}{\lambda_1} \max_{w} ||\nabla g(w)||_{\infty} \leq \frac{d}{\lambda}\cdot \frac{2}{n} = \frac{2d}{n\lambda}
\end{equation*}

Therefore, the $L_1$-sensitivity of $w$ is $\frac{2d}{n\lambda}$.
The Laplace noise with the privacy budget of $\varepsilon_{11}$ is adopted to perturb the true weights,
thus the privacy model training satisfies $\varepsilon_{11}$-DP.
    
\end{proof}

\mypara{Proof 2: Private score calculation satisfies $\varepsilon_{12}$-DP}

\begin{proof}
    \label{proof:score_calculation}
    Different samples are independent of each other, thus the same privacy budget $\varepsilon_{12}$ can be used to add Laplace noise to their propensity scores. 
    According to the parallel composition introduced in~\autoref{section:preliminary_dp}, the step of private score calculation satisfies $\varepsilon_{12}$-DP.

\end{proof}

\mypara{Proof 3: Similar sample matching satisfies $\varepsilon_{2}$-DP}

\begin{proof}
    \label{proof:similar_sample_matching}
    For treatment $T$, the
    random response mechanism is applied to protect the sample's privacy.
    According to~\autoref{section:preliminary_dp}, the random response mechanism meets the requirements of DP.
    The distance calculation between various samples and distance sorting are based on the perturbed $T'$ and $e'(X)$.
    According to the post-processing property of DP,
    these steps do not incur additional privacy loss.
    Following the sequential composition of DP,
    the similar sample matching phase satisfies $\varepsilon_2$-DP.

\end{proof}

\mypara{Proof 4: Causal effect estimation satisfies $\varepsilon_{3}$-DP}

\begin{proof}
    \label{proof:causal_effect_estimation}

    The matching limit is calculated without touching the true data,
    thus it meets DP.
    The counterfactual outcome is perturbed by Laplace noise with the privacy budget of $\varepsilon_3$.
    The ATE estimation is finished based on the noisy outcomes.
    According to the sequential composition and post-processing of DP,
    causal effect estimation satisfies $\varepsilon_{3}$-DP.

\end{proof}

\mypara{Overall Privacy Budget}
According to the above proofs,
in the first phase,
the private regression model training satisfies $\varepsilon_{11}$-DP,
and the private propensity score calculation satisfies $\varepsilon_{12}$-DP.
In the second phase,
the similar sample matching satisfies $\varepsilon_2$-DP.
In the third phase,
the causal effect estimation satisfies $\varepsilon_3$-DP.
Based on the sequential composition of DP,
we obtain that if $l$ is set to sample-level, \autoref{algorithm:put_together} satisfies $\varepsilon$-Sample DP,
where $\varepsilon=\varepsilon_{11}+\varepsilon_{12}+\varepsilon_{2}+\varepsilon_{3}$.

\subsection{Proof of \autoref{throrem:error_bound_label}}
\label{proof_error_bound_label}

\begin{proof}
    Here, we provide the derivation for the outcome sum of the treated group $S_1$ (the derivation of the control group $S_0$ is similar).
    According to~\autoref{eq:expected_squared_error},
    let $\hat S$ denote the estimation of $S$,
    the expected square error $\mathbb{E}[(\hat S-S)^2]$ can be written as the summation of variance and the squared bias of $\hat S$, \ie, $\mathbb{E}[(\hat S-S)^2]=\mathsf{Var}[\hat S]+\mathsf{Bias}[\hat S]^2$.
    
    The variance of $\hat S$ comes from Laplace noise, given the maximum variation range of the outcome $B$, the matching upper limit for each sample $k$ and the privacy budget $\varepsilon$, we obtain the expected error of variance part is $\mathsf{Var}[\hat S]=2(\frac{(k+1)B}{\varepsilon})^2$. 

    The error of bias part comes from the  matching difference caused by whether the matching upper limit is applied.
    For each sample $j \in \{1,2,...,n_c \}$ in the control group, let 
    $u_j$  
    represents the number of times sample $j$ is selected as a neighbor in the original nearest neighbor matching.
    Then, we can obtain the total number of times $R=\sum_{j=1}^{n_c}\max(0,u_j-k)$ that neighbor samples are replaced when matching without the matching upper limit and when matching with the matching upper limit.
    Further, we obtain $\mathsf{Bias}[\hat S]^2=|\mathbb{E}(\hat S)-S|^2\leq(\frac{R}{N}B)^2$, where $N$ is the number of neighbors for each sample in the matching.

    Based on the above derivation, we can obtain:
    \begin{equation*}
        \mathbb{E}[(\hat S-S)^2]\leq2(\frac{(k+1)B}{\varepsilon})^2+(\frac{R}{N}B)^2
    \end{equation*}
\end{proof}

\subsection{Complexity Analysis}
\label{app:complexity_analysis}

In this section, we analyze the computational complexity of various methods, and quantitatively evaluate their running time and memory consumption. 

\mypara{Time Complexity}
We provide the time complexity by analyzing each phase of the algorithms. 
The number of samples is $n$, and the number of covariates is $d$.

For label-level privacy of \method,
the goal of the first phase is to train a logistic regression model.
The time complexity of model training is $\mathcal{O}(nd)$.
In the second phase, we need to calculate and sort the distance between the sample and other samples in the opposite treatment group,
$\mathcal{O}(n\log n)$.
In the third phase,
the counterfactual outcome of each sample is estimated and the potential outcomes are aggregated to compute the final ATE.
The time complexity is $\mathcal{O}(nN)$, where $N$ is the number of neighbors in the counterfactual estimation.
Above all, 
we obtain the total time complexity is $\mathcal{O}(nd+n\log n+nN)$.
For sample-level privacy of~\method,
additional noise injection will incur a time complexity of $\mathcal{O}(n)<\mathcal{O}(n\log n)$.
Therefore, the time complexity of sample-level privacy is also %
$\mathcal{O}(nd+n\log n+nN)$.

For \ipw, the time complexity of propensity score model training is $\mathcal{O}(n_1d)$, where $n_1$ is the number of samples used to train the model. 
The time complexity of differentially private ATE estimation phase is $\mathcal{O}(n_2d)$, where $n_2=n-n_1$ is the number of samples used to estimate ATE.
Since $\mathcal{O}(n_1d)+\mathcal{O}(n_2d)<\mathcal{O}(nd)$, the total time complexity of \ipw is $\mathcal{O}(nd)$.

For \smdp, the time complexity of matching is $\mathcal{O}(n)$,
and the time complexity of smooth sensitivity calculation is $\mathcal{O}(ng)$, where $g$ is the number of different discrete covariate combinations.
Therefore, the time complexity of \smdp is $\mathcal{O}(ng)$.

For \dpci, the time complexity of nuisance model fitting is typically $\mathcal{O}(nd)$.
The time complexity of ATE calculation is $\mathcal{O}(n)$.
Therefore, the total time complexity of \dpci is $\mathcal{O}(nd)$.

For \syn,
in the marginal selection step,
there are
$p=\frac{d(d-1)}{2}$ possible pairwise marginals.
In the $i$-th iteration of marginal selection algorithm,
$(p-i)$ pairwise marginals need to be checked;
thus the time complexity is $\sum_{i=1}^{m}(p-i)=mp-\frac{m(m+1)}{2}=\mathcal{O}(md^2)$,
where $m$ is the number of marginals.
In the dataset generation step,
\syn should go through all marginals $r$ times to ensure consistency.
Thus, the time complexity is $mr$,
and the overall time complexity is $\mathcal{O}(md^2+mr)$.

For \aim, all 1-way marginals are measured in the beginning,
and the time complexity is $\mathcal{O}(nd)$.
In the iterative stage,
the noise is injected into the selected $v$-way marginals in each round.
The corresponding time complexity is $\mathcal{O}(Tqnv)$, where $T$ is the number of iteration, and $q$ is the number of candidate queries in each round.
Therefore, the total time complexity of \aim is $\mathcal{O}(nd+Tqnv)$.

\mypara{Space Complexity}
For \method, it requires storing the covariate matrix and the propensity score, 
thus the space complexity of \method is 
$\mathcal{O}(nd)$.
For \ipw, the memory consumption mainly comes from the original dataset and model training,
and the corresponding space complexity is $\mathcal{O}(nd)$.
For \smdp, it requires storing the original dataset and the grouping information, thus the space time complexity is $\mathcal{O}(nd+g)$.
For \dpci, it mainly requires storing the dataset and nuisance models, thus the space complexity is $\mathcal{O}(nd)$.
For \syn, the memory consumption consist of two parts, \ie, marginal tables and synthetic
dataset.
The memory comsumption of marginal tables is the product of the number of marginals $m$ and the average number of cells for each marginal $C$,
and the memory consumption of synthetic dataset is $\mathcal{O}(nd)$.
Therefore, the space complexity of \syn is $\mathcal{O}(mC+nd)$.
For \aim,
it requires to store the original dataset with the space complexity of $\mathcal{O}(nd)$.
In the iteration process,
\aim needs to store the noise measurement values.
The space complexity is $\mathcal{O}(Th^v)$, where $h^v$ is the domain size for any $v$-way marginal. 
In addition, \aim also needs to storage the Private-PGM model parameters,
with the space complexity of $\mathcal{O}(S)$,
where $S$ is the junction tree size.
Therefore, the total space complexity of \aim is $\mathcal{O}(nd+Th^v+S)$.

\mypara{Empirical Evaluation}
\autoref{table:comparsion_running_time} and \autoref{table:comparsion_memory_consumption} show  the
running time and the memory consumption for all methods on the four datasets (see their details in \autoref{table:dataset_statistics}).

The empirical running time in \autoref{table:comparsion_running_time}
illustrates that the performance of differentially private ATE estimation methods is better than synthesis-based methods.
The running time of \ipw, \smdp and \method on the four datasets is smaller than $1$s, which reflects their high efficiency.
\method incurs longer runtime on the ACIC dataset due to its larger size, which increases the computational cost of both model training and distance sorting.
The running time of sample-level of \method is slightly higher than label-level since sample-level requires additional noise injection.
In addition, the running time of \syn is significantly longer than \method since it requires capturing the features of many marginals.
\aim exhibits the longest time due to a large number of iterations.

\autoref{table:comparsion_memory_consumption} shows the memory consumption.
The consumption of \method is the 
lowest.
We find that the consumption of \method, \ipw and \smdp is close because most of the same memory is used to store the original dataset.
The consumption of \aim is significantly higher than that of other methods since it requires to storage the selected $v$-way marginals and the junction tree.

\begin{table}[!t]
\caption{Comparison of computational complexity.}
    \centering
    \setlength{\tabcolsep}{0.8em}
    \begin{tabular}{c | c | c }
    \toprule
    \textbf{Methods} & 
    \textbf{Time} & 
    \textbf{Space}  
    \\
    \midrule
      \ipw  & $\mathcal{O}(nd)$ & $\mathcal{O}(nd)$
       \\
       \smdp  & $\mathcal{O}(ng)$ & $\mathcal{O}(nd+g)$
       \\
       \dpci & $\mathcal{O}(nd)$ & $\mathcal{O}(nd)$
       \\
       \syn  & $\mathcal{O}(md^2+mr)$ & $\mathcal{O}(mC+nd)$
       \\
       \aim & $\mathcal{O}(nd+Tqnv)$  & $\mathcal{O}(nd+Th^v+S)$
       \\
       \method (sample) 
       & %
       $\mathcal{O}(nd+n\log n+nN)$
       & $\mathcal{O}(nd)$ 
       \\
       
       \method (label)
       & %
       $\mathcal{O}(nd+n\log n+nN)$
       & $\mathcal{O}(nd)$ 
       \\
    \bottomrule
    \end{tabular}
    \label{table:comparsion_computation_complexity}
\end{table}

\begin{table}[!t]
\caption{Comparison of running time (measured by seconds).}
    \centering
    \setlength{\tabcolsep}{0.6em}
    \begin{tabular}{c | c | c| c | c }
    \toprule
    & \multicolumn{4}{c}{\textbf{Datasets}} \\
    {\textbf{Methods}} & IHDP  & Lalonde  &  ACIC &  Synth\\
     \midrule
        {\ipw}  & 0.29s & 0.16s  &  0.32s & 0.06s   \\
        {\smdp}  & 0.03s & 0.04s  &  0.06s & 0.02s   \\
        {\dpci}  & 0.09s & 0.04s  &  0.94s & 0.08s   \\
       {\syn}  & 67.02s & 5.18s  &  4861.72s & 8.11s   \\
       {\aim}  & 210.35s &  22.24s & 12569.72s  &  104.63s  \\
       {\method (sample)} & 0.08s & 0.04s & 0.68s & 0.05s  \\
       {\method (label)} & 0.06s  & 0.02s  & 0.55s  & 0.04s  \\
      \bottomrule
    \end{tabular}
    
    \label{table:comparsion_running_time}
\end{table}

\begin{table}[!t]
\caption{Comparison of memory consumption (measured by Megabytes).}

    \centering
    \setlength{\tabcolsep}{0.6em}
    \begin{tabular}{c | c | c | c | c }
    \toprule
    & \multicolumn{4}{c}{\textbf{Datasets}} \\
    {\textbf{Methods}} &  IHDP &   Lalonde &  ACIC &  Synth\\
     \midrule
       {\ipw}  & 471.13 &  470.92 &  493.67 &  472.70   \\
       {\smdp}  & 469.13 &  468.92 &  477.31 &  470.81   \\
       {\dpci} & 479.84 & 476.26 & 483.60 & 476.05   \\
       {\syn}  & 540.44 &  534.98 &  572.21 &  537.60   \\
       {\aim}  & 1712.74 & 933.76 &  4177.00 &  1043.12   \\
       {\method (sample)} & 470.60 & 469.97 & 484.55 & 470.92  \\
       {\method (label)} & 468.19 & 467.46 & 479.83 & 468.29  \\
      \bottomrule
    \end{tabular}
    \label{table:comparsion_memory_consumption}
\end{table}

\subsection{Dataset Description}

\label{appendix_datasets}
The details of the four datasets are as follows.

\begin{itemize}
\item \textbf{IHDP~\cite{hill2011bayesian}.} 
The Infant Health
and Development Program (IHDP) dataset is a semi-real dataset,
where only the outcome value is simulated.
The IHDP dataset is  constructed by combining real-world covariates from a randomized experiment with simulated assignments and outcomes. 
To create a challenging observational study, a selection bias was artificially introduced by removing a non-random subset (all non-white infants) from the treated group. 
The outcome variable is then synthetically generated based on predefined response surfaces, which model the relationship between covariates and the treatment effect.
This dataset includes $n=747$ individuals comprised of $n_t=139$ treated and $n_c=608$ control individuals.
The treatment involves specialist home visits for children, while the outcome is their future cognitive test scores.
The dimension of the covariates is $d=25$, including demographic information, infant health, socioeconomic status, \etc

\item \textbf{Lalonde~\cite{dehejia1999causal}.} 
The Lalonde dataset is a real dataset that comes from an evaluation study of the national supported work (NSW) program.
In particular, the experimental treated group from the NSW study (which received job training) is combined with a non-experimental control group drawn from observational surveys.
This dataset is composed of $n=445$ individuals, where $n_t=185$ individuals belong to the treated group and $n_c=260$ individuals belong to the control group. 
The treatment refers to whether or not to participate in the NSW program,
and the outcome is earning in
1978.
This dataset includes $d=8$ dimensions of covariates, such as the age, years of education, income in previous years, \etc

\item \textbf{ACIC~\cite{dorie2019automated}.}
ACIC comes from the Atlantic Causal Inference Conference competition in 2016 for causal challenges. 
The dataset used in this competition is semi-real, \ie, the covariates are real, while the treatment and outcome are synthetic. 
The treatment is assigned based on complex, pre-defined propensity score models, and the outcome is generated by response functions that incorporate heterogeneous treatment effects.
The covariates of ACIC come from the Collaborative Perinatal Project, which provides real medical data. Each sample contains $d=58$ covariates, including binary, categorical, and continuous variables. 
Treatment and outcome are generated by the organizer. 
In our experiment, there are a total of $n=4802$ samples, of which $n_t=858$ samples are in the treated group and $n_c=3944$ samples are in the control group.

\item \textbf{Synth~\cite{imbens2015causal}.}
This dataset is a synthetic dataset. 
We adopt the simulated method introduced in~\cite{koga2024differentially} to generate this dataset.
First, we simulate a covariate matrix.
A total of $n=1000$ samples are generated, each characterized by $d=20$ covariates.
These covariates are independently sampled from a uniform distribution.
To mimic the selection bias inherent in observational data, the treatment assignment 
$T$ is made dependent on the covariates.
We employ a logistic model where the propensity score for each sample is given by $e(X_i)=sigmoid(a\cdot(2X_i-1))$.
The parameter $a$, which controls the extent of selection bias, is drawn from the uniform distribution.
Next, the observed outcome 
$Y$ for each sample is synthesized using a linear response function: $Y_i=b \cdot X_i + \tau \cdot T_i+q_i$. Here, $b$ is a vector of coefficients sampled from a uniform distribution, representing the heterogeneous influence of each covariate on the outcome.
The scalar $\tau$ is set to $0.5$, which stands for the ATE estimate.
 $q_i$ is an independent noise term, also sampled from a uniform distribution, which introduces random variability into the outcome.
 After the above process, the treatment assignment yields a naturally imbalanced split, resulting in a final dataset with $n_t=489$ samples in the treated group and $n_c=511$ samples in the control group.

\end{itemize}

\section{Additional Results}
\label{section:additional_results}

\subsection{End-to-end Evaluation under Various $N$}
\label{subsec:appendix_end_to_end}

\autoref{fig:appendix_end_to_end} illustrates the relative errors under various numbers of matched neighbors in the counterfactual estimation.
We observe that
\method shows a consistent and similar performance,
which illustrates that \method is robust to the variation of $N$.

\begin{figure*}[!t]
\captionsetup[ploture]{labelformat = parens, labelsep = space, font = small,labelfont=bf}
\centering
\includegraphics[width=0.95\textwidth]{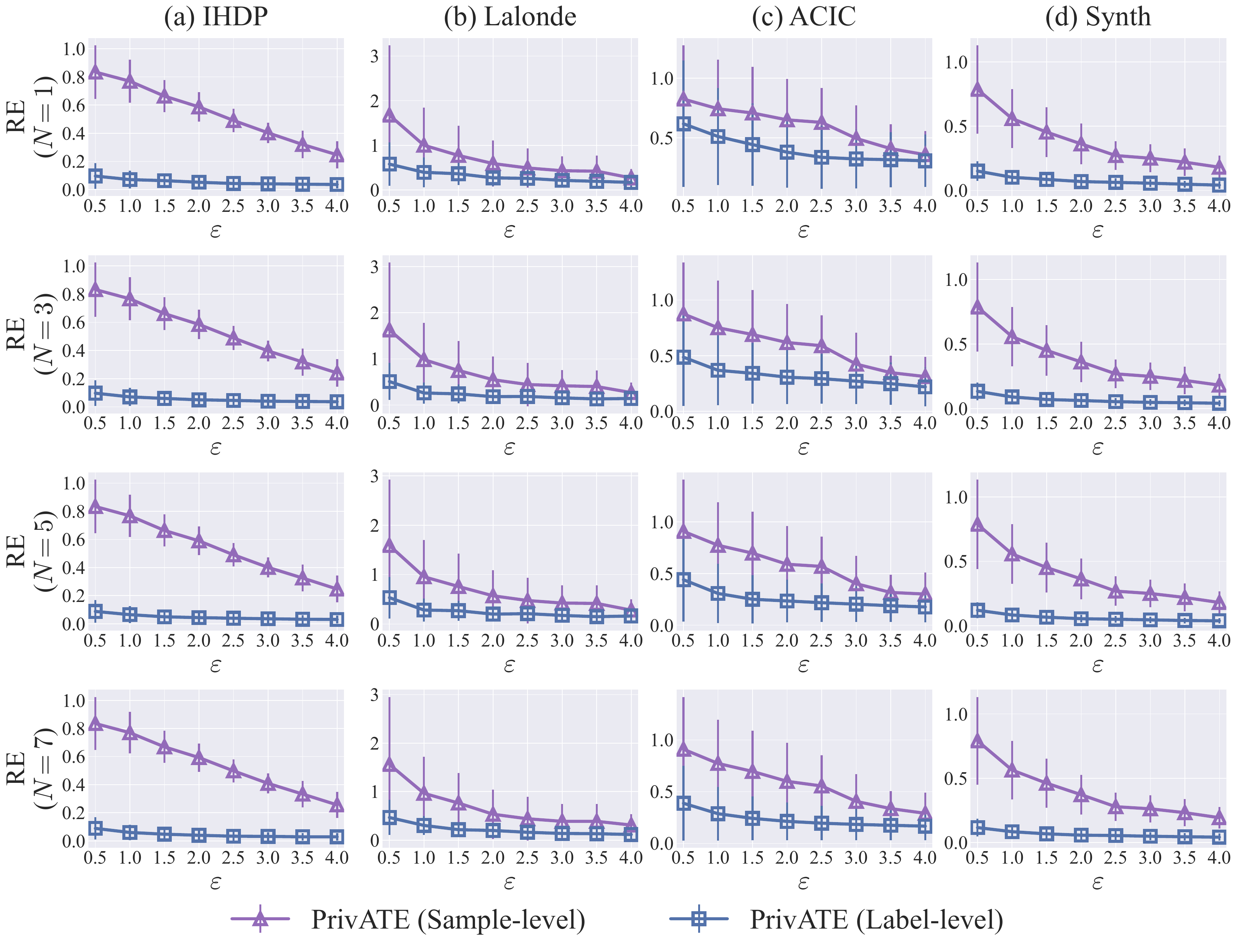}
\vspace{-0.1cm}
\caption{
{
End-to-end comparison of various numbers of matched neighbors.
The columns represent the used datasets, and the rows stand for different number of matched neighbors in the counterfactual estimation.
In each plot, the x-axis denotes the privacy budget $\varepsilon$, and the y-axis denotes the relative error.}}
\label{fig:appendix_end_to_end}
\end{figure*}

\subsection{Parameter Variation for Label-level Privacy}
\label{subsec:appendix_param_label}

\mypara{Choice of Matching Limit}
\autoref{fig:appendix_compare_fixed_k} illustrates 
the performance of various matching limit determination mechanisms for label-level  privacy under various $N$.
Similar to the analysis for choice of matching limit in~\autoref{section:paramter_study_label_level},
we find that the fixed value approach is difficult to obtain great performance on all datasets.
The adaptive determination mechanism of~\method can achieve low REs in most cases.

\begin{figure*}[!t]
    \centering
    \includegraphics[width=0.93\textwidth]{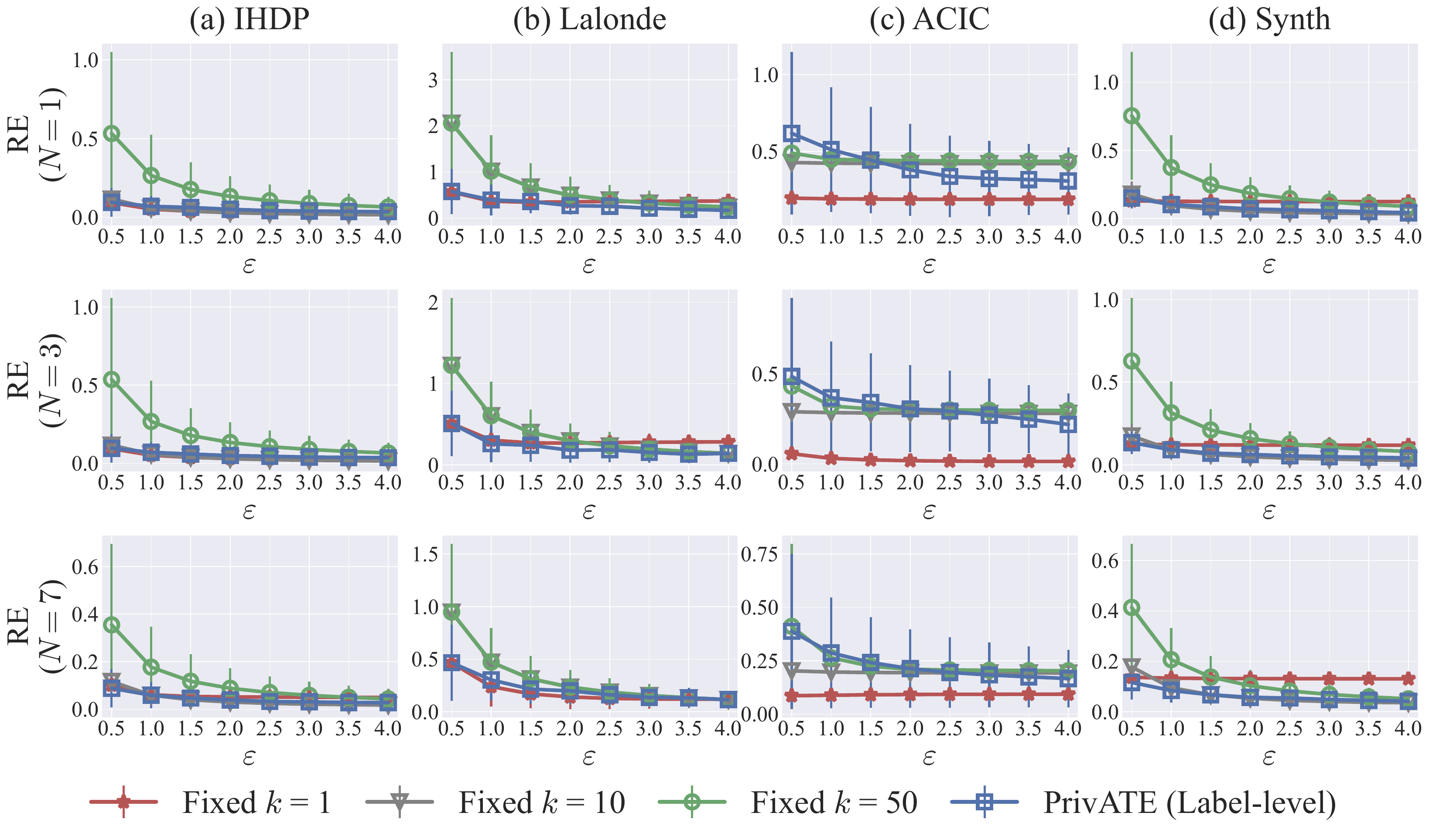}
    \vspace{-0.35cm}
    \caption{
    Impact of different matching limit determination mechanisms in the label-level privacy of \method. 
    The columns represent the used datasets, and the rows stand for different number of matched neighbors in the counterfactual estimation. 
    In each plot, the x-axis denotes the privacy budget $\varepsilon$, and the y-axis denotes relative error. 
    }
    \vspace{-0.25cm}
    \label{fig:appendix_compare_fixed_k}
\end{figure*}

\mypara{Impact of Error Coefficient in Matching Limit Calculation}
\autoref{fig:appendix_vary_label_c} illustrates 
the performance of various error coefficients in the label-level setting of~\method under various $N$.
Similar to the results in~\autoref{section:paramter_study_label_level},
we observe that $c=0.01$ can achieve promising performance across various settings.

\begin{figure*}[!t]
    \centering
    \includegraphics[width=0.93\textwidth]{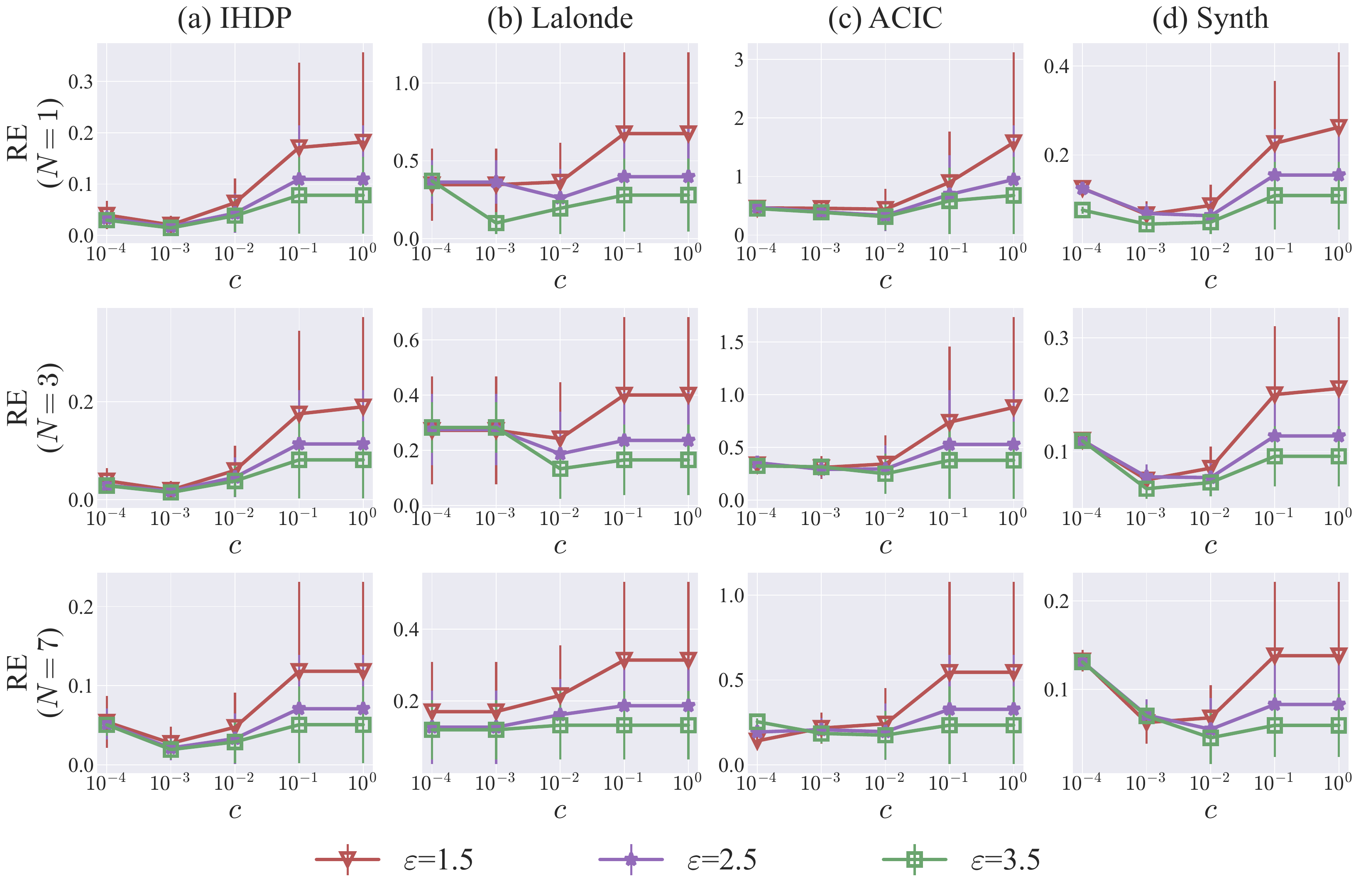}
    \vspace{-0.35cm}
    \caption{
    Impact of different error coefficients $c$ in the label-level privacy of \method. 
    The columns represent the used datasets, and the rows stand for different number of matched neighbors in the counterfactual estimation. 
    In each plot, 
    the x-axis denotes the error coefficient $c$,
    and the y-axis denotes relative error. 
    }
    \vspace{-0.25cm}
\label{fig:appendix_vary_label_c}
\end{figure*}

\subsection{Parameter Variation for Sample-level Privacy}
\label{subsec:appendix_param_sample}

\mypara{Choice of Matching Limit}
\autoref{fig:appendix_compare_fixed_k_sample} illustrates 
the performance of various matching limit determination mechanisms for sample-level privacy under various $N$.
Similar to the observation in~\autoref{section:paramter_study_sample_level},
we find that larger $k$ tends to cause significant errors due to inaccurate matching and obvious perturbation noise.
The adaptive determination mechanism of~\method can achieve promising performance in most cases.

\begin{figure*}[!t]
    \centering
    \includegraphics[width=0.93\textwidth]{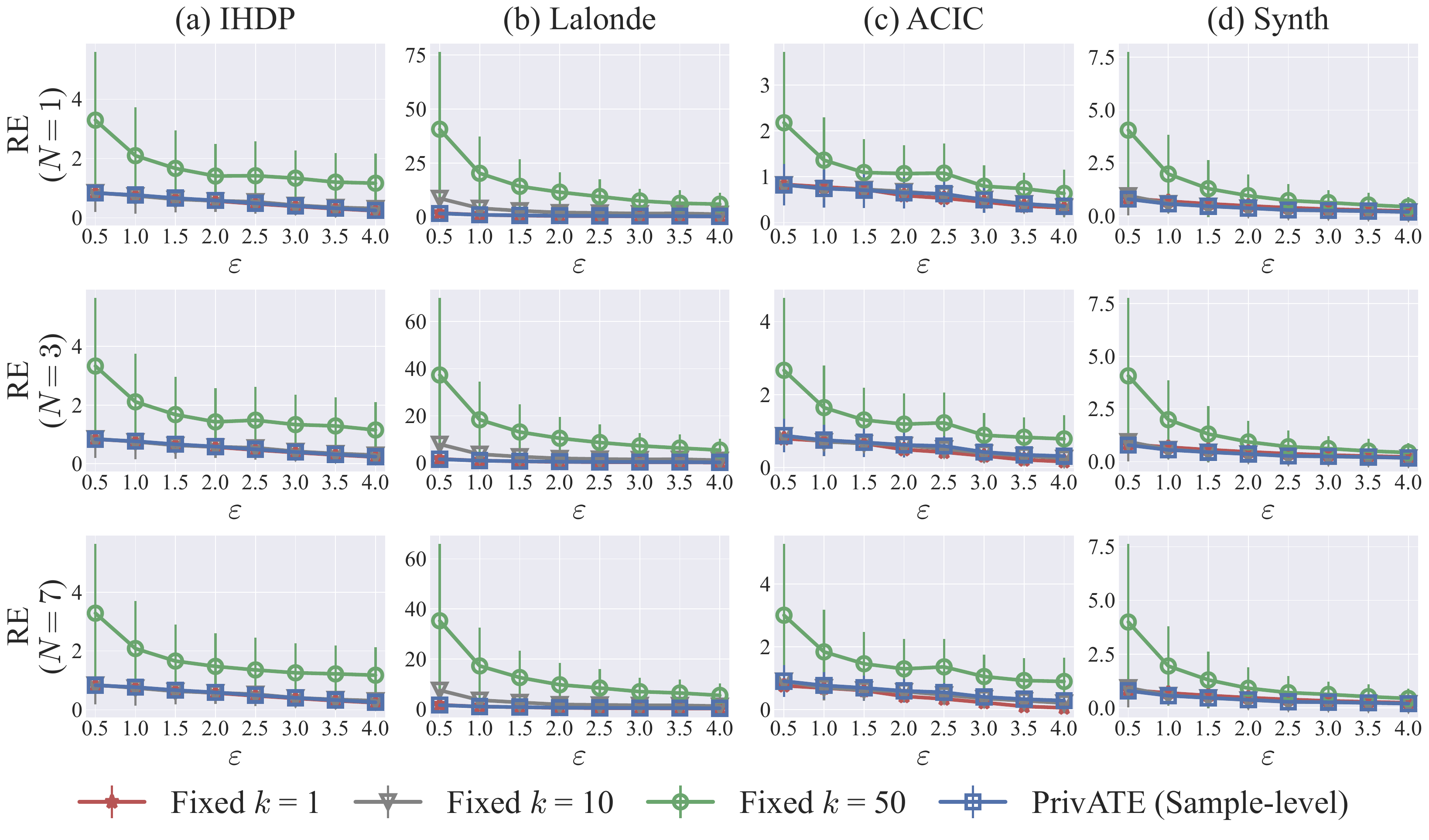}
    \vspace{-0.4cm}
    \caption{
    Impact of different matching limit determination mechanisms in the sample-level privacy of \method. 
    The columns represent the used datasets, and the rows stand for different number of matched neighbors in the counterfactual estimation. 
    In each plot, the x-axis denotes the privacy budget $\varepsilon$, and the y-axis denotes relative error. 
    }
    \vspace{-0.1cm}
    \label{fig:appendix_compare_fixed_k_sample}
\end{figure*}

\mypara{Impact of Error Coefficient in Matching Limit Calculation}
\autoref{fig:appendix_vary_user_h} provides 
the performance of various error coefficients in the sample-level setting of~\method under various $N$.
We can obtain the similar results to~\autoref{section:paramter_study_sample_level}.

\begin{figure*}[!t]
    \centering
    \includegraphics[width=0.93\textwidth]{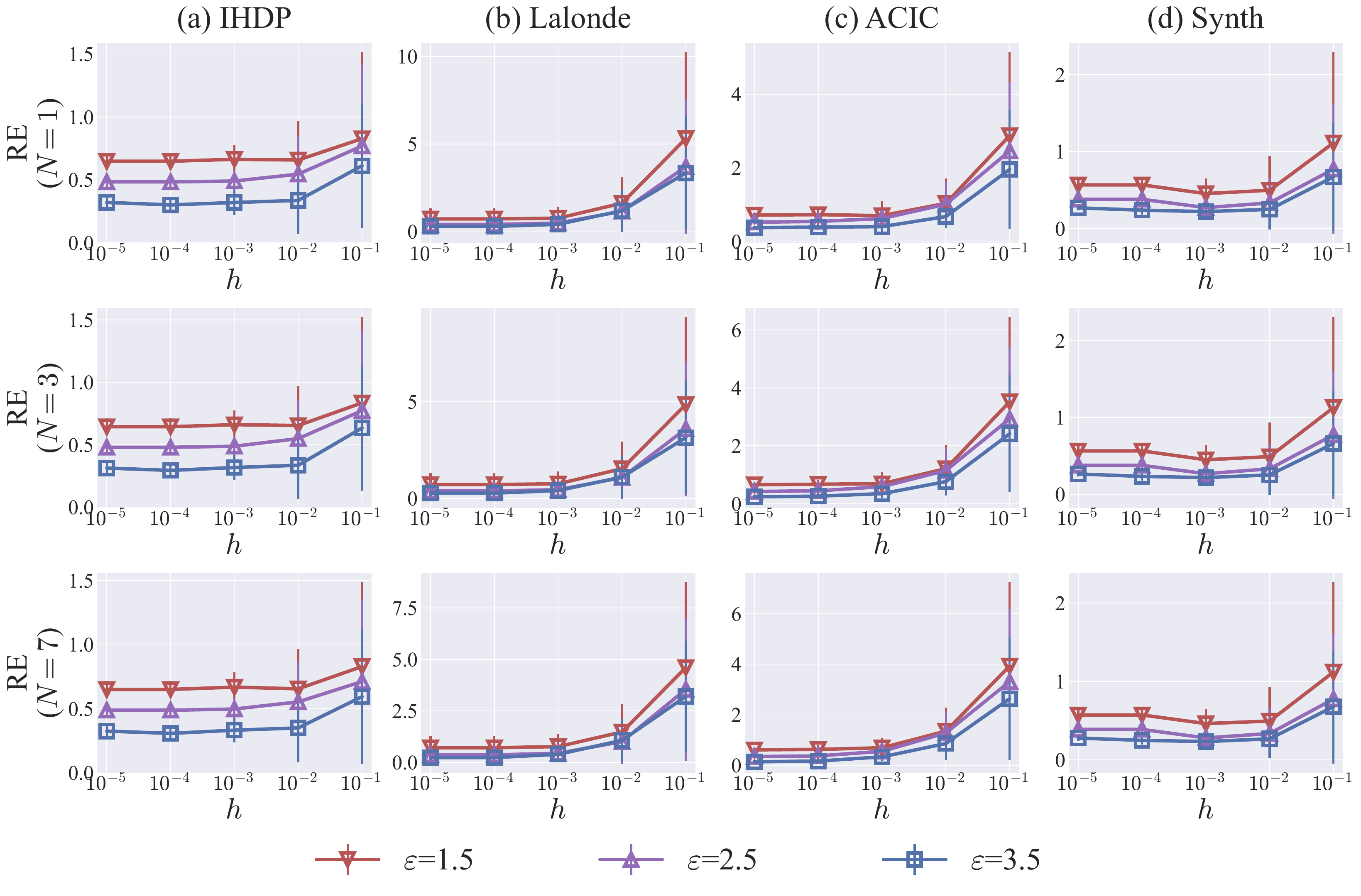}
    \vspace{-0.4cm}
    \caption{
    Impact of different error coefficients $h$ in the sample-level privacy of \method. 
    The columns represent the used datasets,
    and the rows stand for different number of matched neighbors in the counterfactual estimation. 
    In each plot, 
    the x-axis denotes the error coefficient $h$,
    and the y-axis denotes relative error. 
    }
    \vspace{-0.1cm}
    \label{fig:appendix_vary_user_h}
\end{figure*}

\mypara{Impact of Privacy Budget Allocation}
Recalling the sample-level setting in~\autoref{section:putting_things_together}, \method divides the entire privacy budget into three phases. 
In the section, we evaluate the impact of various privacy budget allocation on the four datasets.
\autoref{fig:vary_privacy_budget_allocation}
illustrates the REs of ATE estimation for different budget allocation strategies. 
We find that the allocation ratio has an obvious influence on the ATE estimate.
The 
observations are as follows.

\begin{figure*}[!t]
    \centering
    \includegraphics[width=0.96\textwidth]{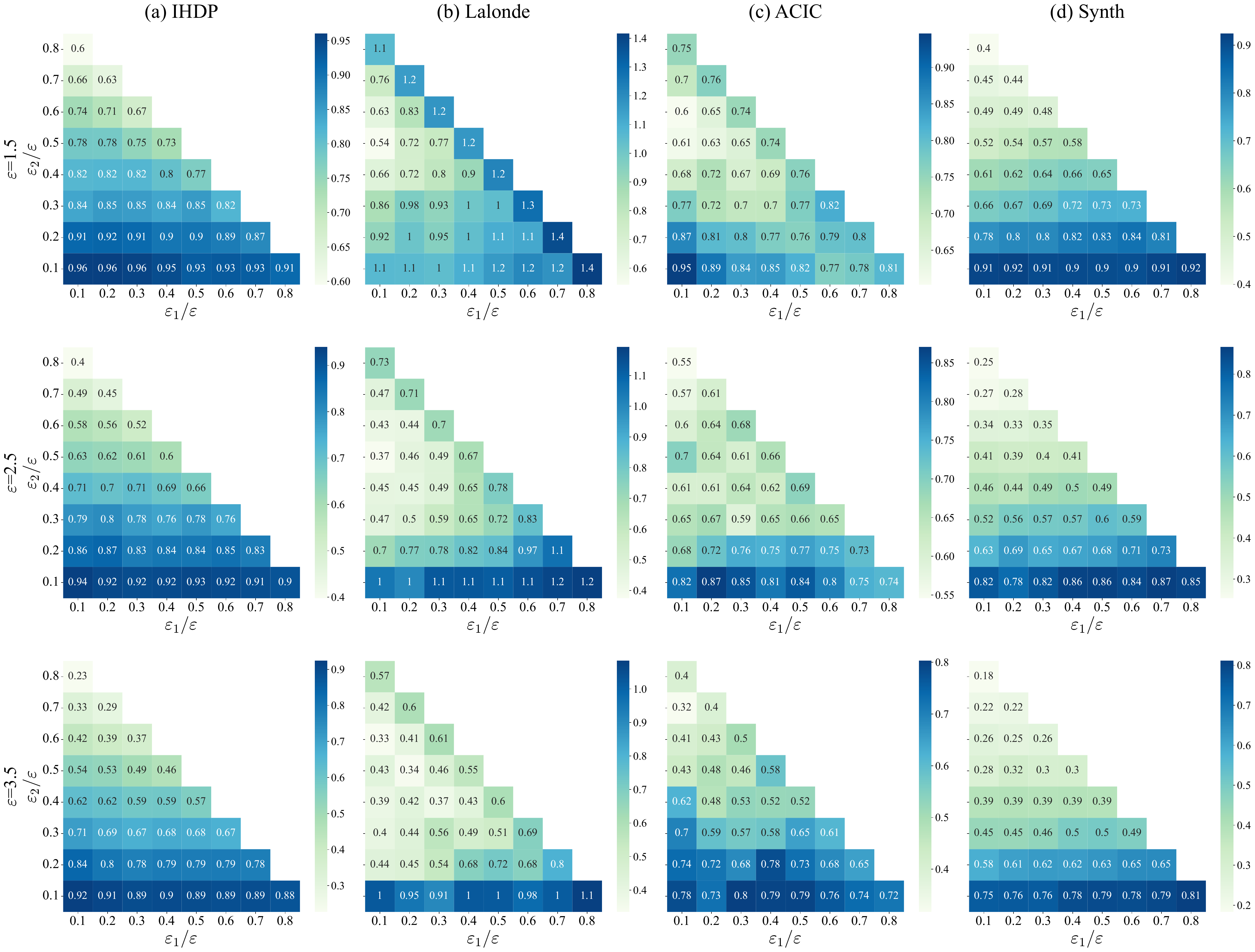}
    \vspace{-0.1cm}
    \caption{
    Impact of different privacy budget allocation in the sample-level privacy of \method when the number of matched neighbors $N$ is 5.
    The columns represent the used datasets, and the rows stand for different privacy budgets. 
    In each plot, the x-axis denotes the ratio of $\varepsilon_1$, the y-axis denotes the ratio of $\varepsilon_2$ ($\varepsilon_3=\varepsilon-\varepsilon_1-\varepsilon_2$), and the values in the grid denotes the relative error.} 
    
    \label{fig:vary_privacy_budget_allocation}
\end{figure*}

First, a small privacy budget ratio in the second phase will produce a high RE.
$\varepsilon_2$ is used to perturb  true treatment. 
If the proportion of $\varepsilon_2$ is small, the relevant estimates will exhibit obvious deviations when performing matching and counterfactual calculations.
Therefore, more privacy budget needs to be allocated to the second phase to achieve better estimation results.
In addition, the optimal division ratios under different privacy budgets are also various.
When the total privacy budget is small (\ie, $\varepsilon=1.5$), 
it is difficult to achieve great results in both regression model training and sample matching. 
In this case, increasing the ratio of $\varepsilon_1$ or $\varepsilon_3$ usually cannot effectively reduce the RE of ATE estimate.
As the privacy budget increases,
the role of each part becomes more obvious,
and the ratios of each phase needs to be carefully allocated to achieve a low RE.
Moreover, we also observe that the optimal privacy budget allocation ratios for different datasets are also various.
The number of samples in the Lalonde dataset 
is small and the range of outcome is large,
while the dimension of ACIC is quite high.
Therefore, it is significant to take into account the effects of three phases to achieve low REs for these two datasets.

Generally, we find that there is no optimal privacy budget allocation for all datasets and privacy budgets.
Considering the role of different phases,
we choose to set $\varepsilon_1:\varepsilon_2:\varepsilon_3=0.1:0.7:0.2$ for all experiments.
This setting shows low REs on various datasets and privacy budgets.

\end{document}